\documentclass[11pt,draftcls, onecolumn]{IEEEtran}

\usepackage{amsmath}
\usepackage{amsthm}
\usepackage{amssymb}
\usepackage{verbatim}
\usepackage{algorithm}
\usepackage{algorithmic}
\usepackage{graphicx}
\usepackage{epstopdf}
\usepackage{paralist}
\usepackage{color}
\usepackage{booktabs}
\usepackage{extarrows}
\usepackage{chemarrow}
\usepackage{cite}

\ifCLASSOPTIONcompsoc
\usepackage[tight,normalsize,sf,SF]{subfigure}
\else
\usepackage[tight,footnotesize]{subfigure}
\fi

\theoremstyle{plain}
\newtheorem{lemma}{Lemma}

\newtheorem{theorem}{Theorem}

\newtheorem{definition}{Definition}

\newtheorem{proposition}{Preposition}

\theoremstyle{definition}

\floatname{algorithm}{Algorithm}

\newcommand{\argmax}{\operatornamewithlimits{argmax}}

\allowdisplaybreaks[4]

\begin{document}
%
\title{Pricing Mobile Data Offloading: A Distributed Market Framework}

\author{~~~~~~~~Kehao~Wang,~~ Francis C. M. Lau, \textit{Senior Member, IEEE},~~ Lin Chen, \textit{Member, IEEE}, ~~~~~~~~~~~~~~~~~~~~~~and Robert Schober, \textit{Fellow, IEEE}
\IEEEcompsocitemizethanks{\IEEEcompsocthanksitem

K.~Wang and Francis C. M. Lau are with the Department of Electronic and Information Engineering, the Hong Kong Polytechnic University, Kowloon, Hong Kong. (email: www.kehao@gmail.com, encmlau@polyu.edu.hk).

L.~Chen is with the Laboratoire de Recherche en Informatique (LRI), Department of Computer Science, the University of Paris-Sud, 91405 Orsay, France (e-mail:lin.chen@lri.fr).

Robert Schober is with the Institute for Digital Communications (IDC), Friedrich-Alexander-University Erlangen-Nurnberg (FAU), Germany (email:schober@lnt.de).
}}

\IEEEcompsoctitleabstractindextext{%
\begin{abstract}

Mobile data offloading is an emerging technology to avoid congestion in cellular networks and improve the level of user satisfaction. In this paper, we develop a distributed market framework to price the offloading service, and conduct a detailed analysis of the incentives for offloading service providers and conflicts arising from the interactions of different participators. Specifically, we formulate a multi-leader multi-follower Stackelberg game (MLMF-SG) to model the interactions between the offloading service providers and the offloading service consumers in the considered market framework, and investigate the cases where the offloading capacity of APs is unlimited and limited, respectively. For the case without capacity limit, we decompose the followers' game of the MLMF-SG (FG-MLMF-SG) into a number of simple follower games (FGs), and prove the existence and uniqueness of the equilibrium of the FGs from which the existence and uniqueness of the FG-MLMF-SG also follows. For the leaders' game of the MLMF-SG, we also prove the existence and uniqueness of the equilibrium. For the case with capacity limit, by considering a symmetric strategy profile, we establish the existence and uniqueness of the equilibrium of the corresponding MLMF-SG, and present a distributed algorithm that allows the leaders to achieve the equilibrium. Finally, extensive numerical experiments demonstrate that the Stackelberg equilibrium is very close to the corresponding social optimum for both considered cases.
\end{abstract}

}

\maketitle

\IEEEdisplaynotcompsoctitleabstractindextext

%
\IEEEpeerreviewmaketitle

\section{Introduction}
\label{section:introduction}

\subsection{Background}
The data traffic in cellular networks has seen a tremendous growth over the past few years due to the explosion of mobile devices, e.g. smart phones, tablets, laptops etc. The increasing data traffic in cellular networks suggests that traffic from cellular networks should be offloaded so as to alleviate traffic congestion and improve user satisfaction. Thus, mobile data offloading emerged as a promising approach to utilize certain complementary transmission technologies to deliver data traffic originally transmitted over cellular networks to the users. Recently, a large number of studies have investigated the potential benefits of mobile data offloading and various innovative schemes have been proposed to better manage data traffic including WiFi~\cite{Lee2010coNext,Kyunghan2013ToN,Xuejun2013TMC,Ristanovic2011Mass,Bennis2013CM}, femtocells~\cite{Andrews2012JSAC,Lingjie2013TMC,HuaiLei2013TMC,Saker2012JSAC,Seyoung2012JSAC}, and opportunistic offloading~\cite{Yongli2013TMC,bohan2013TMC}. In fact, these studies have shown that data offloading is a cost-effective and energy-prudent approach to resolve network congestion and improve network capacity.

\subsection{Motivation}

However, the merit of mobile data offloading does not always guarantee that offloading is adopted by the offloading service providers (OSPs) and offloading service consumers (OSCs), i.e., mobile data flows, in practice. One of the most important reasons for not adopting mobile data offloading is the lack of economic incentives, i.e., OSPs may be reluctant to make their resources available for offloading data traffic without permission or appropriate economic reimbursement since offloading data traffic will consume their limited wireless resources and reduce broadband connection capacity. Thus, it is of significant importance to analyze the economic implications of mobile data offloading from the perspective of both OSPs and OSCs. 
For ease of presentation, in this paper, we focus on WiFi offloading in which OSPs and OSCs represent Access Points (APs) and cellular data flows, respectively.

From an economics point of view, there are many works considering the interaction between APs and cellular data flows. For example,~\cite{Joohyun2013info} studied delayed WiFi offloading by modeling the interactions between APs and cellular data flows as a two-stage sequential Stackelberg game with one leader and multiple followers.
In~\cite{Lingao2013info}, the authors investigated the economics of mobile data offloading through WiFi or femtocells, and utilized a multi-leader multi-follower Stackelberg game (MLMF-SG) to achieve the subgame perfect equilibrium (SPE), and further compared the SPE with the corresponding outcomes in a perfect competition market and in a monopoly market without price participation, respectively.
In~\cite{Iosifidis2013wiopt}, the authors considered the scenario where each of the mobile network operators (MNOs) can employ multiple APs to offload its data traffic and each AP can concurrently serve traffic from different MNOs. The proposed market scheme incurred minimum communication overhead and created non-negative revenue for the market broker without requiring a priori information about MNOs and APs.

Motivated by~\cite{Joohyun2013info,Lingao2013info,Iosifidis2013wiopt}, we consider a typical offloading scenario where a number of cellular data flows offload their data traffic to a number of APs in their vicinity, e.g., hotspots near base stations. In particular, we propose a pricing framework based on the concept of `paying for offloading' to ensure efficient use of the offloading APs. Under this framework each cellular data flow corresponding to a mobile source-destination pair offers a payment to incentivize APs to participate in offloading, and then the payment is shared in proportion to the amount of data offloaded to each AP. Hence, the utility of an AP is its share of received payment minus its own offloading cost. For a cellular data flow, its utility is defined as a generic concave function of the sum of the utilities from offloading on the APs minus the cost paid to these offloading APs. We model the interaction of the APs and the cellular data flows as an MLMF-SG, where the APs are the followers who respond to the payment offered by the cellular data flows (i.e., each AP offloads a part of the data of some flows such that its utility is maximized, given the payment offered by the flows and the actions of its competing peers); and the cellular data flows are the leaders who set the payment to maximize their own utility in anticipation of the Nash equilibrium (NE) response of the followers. Notwithstanding our interest in the mobile data offloading context, the considered model is generic enough to be applied any other scenario where a set of `jobs' compete for the services of a pool of `workers', such that the jobs set their payment rates, workers are free to choose the job they will attempt, and payment from each job is eventually shared according to certain allocation rules among all the workers that serve the job.

Unlike most pricing methods in the existing literature that involve only one type of selfish players~\cite{Joohyun2013info} or two types of selfish players without competition between them~\cite{Lingao2013info,Iosifidis2013wiopt}, our framework features two types of players, each of which competes not only with its peers but also with the players of the other type. This property distinguishes our work from the scenario considered in~\cite{Lingao2013info,Iosifidis2013wiopt}, where only players of the same type can compete with each other although there exist two types of selfish players. This difference cause the utility functions of players in this paper to be completely different from those in~\cite{Lingao2013info,Iosifidis2013wiopt} as far as concavity is concerned. Concretely, with the strategy profile in~\cite{Lingao2013info,Iosifidis2013wiopt}, the utility functions of both followers and leaders are concave, which ensures that there exists an equilibrium in the followers' game and the leaders' game, respectively.
However, in our case, the payment from a flow is shared proportionally among all APs according to the amount of data offloaded to each AP. As a consequence, an AP's utility depends not only on its own strategy but also on the strategies of its peers, which leads to complex interactions among the APs. Accordingly, the sharing of payment causes the utility functions to be non-concave, which necessitates a completely new and original study of the game's equilibrium.


\subsection{Contributions}
The main contributions of this paper can be summarized as follows:
\begin{itemize}
  \item We develop a distributed market pricing framework for mobile data flows to price the offloading service.
  \item We formulate a Stackelberg game to model the interactions between offloading service providers and offloading service consumers under the market framework, and investigate the cases where the offloading capacity of APs is limited and unlimited, respectively. For both cases, we establish the existence and uniqueness of the equilibrium of the proposed Stackelberg game, obtain the Stackelberg equilibrium in closed form when the offloading capacity of the APs is not limited, and further propose a distributed pricing algorithm to ensure that the game converges to an equilibrium when the offloading capacity of the APs is limited.
  \item We conduct a large number of simulations to verify our theoretical analysis on the proposed Stackelberg game for the two considered cases. As a noteworthy property of the developed framework, simulation results demonstrate that the Stackelberg equilibrium is very close to the social optimum.

\end{itemize}

\subsection{Related Work}
\label{sec:rw}

To provide better service in cellular networks, a body of literature has proposed to exploit various kinds of technologies to offload data traffic. These works adopt three different approaches for offloading.

The first approach is opportunistic offloading which utilizes opportunistic communication to offload cellular traffic. For example, the authors of~\cite{Yongli2013TMC} considered the heterogeneities of mobile data and mobile users in realistic disruption tolerant networks, and established a mathematical framework to study the problem of multiple-type mobile data offloading. 
In~\cite{bohan2013TMC}, opportunistic communication was exploited to facilitate information dissemination in the emerging mobile social networks and to reduce the amount of mobile data traffic.

The second approach is femtocell offloading which has emerged as another primary option for macrocellular data offloading. In~\cite{Andrews2012JSAC}, the potential benefits and costs of deploying femtocells were surveyed.
In~\cite{Lingjie2013TMC}, the authors investigated the network operator's profit gain from offering dual services through both macrocells and femtocells.
The authors of~\cite{HuaiLei2013TMC} considered the tradeoff between reducing the paging cost in mobility management and registration signaling overhead, and proposed a delay registration algorithm that postpones the registration and reduces signaling overhead while sustaining the traffic offloading capability of the femtocell.
In~\cite{Saker2012JSAC}, optimal sleep/wake up schemes were studied for the base stations of network-operated femtocells to offload part of its traffic to minimize the energy consumption of the overall heterogeneous network while preserving quality of service (QoS). 
In~\cite{Seyoung2012JSAC}, the authors studied the economic aspects of femtocell services for the case of a monopoly market.
The authors of~\cite{Tehrani2012icc} proposed a dynamic pricing scheme based on market equilibrium and non-cooperative game such that the mobile service providers can gain more revenue than with a fixed pricing scheme.
In~\cite{Jiming2013CC}, the authors focused on the inter-femtocell interference in three-dimension scenarios, and classified multiple femtocells into a number of groups according to the amount of interference caused to others.

In this paper, we focus on another approach for mobile data offloading which exploits the freely available WiFi networks, and is referred to as WiFi offloading.
In~\cite{Kyunghan2013ToN}, the authors first presented a quantitative study for the performance of 3G mobile data offloading through WiFi networks, and then proposed a distribution model-based simulator to investigate the average performance of offloading for a given WiFi deployment condition.
In~\cite{Xuejun2013TMC}, the authors studied the tradeoff between the amount of traffic being offloaded and the user satisfaction, and provided an incentive framework based on reverse auction to motivate users to leverage their delay tolerance for cellular traffic offloading.
This performance gain can be improved by delaying transmission~\cite{Xuejun2013TMC} and predicting WiFi availability~\cite{Xuejun2013TMC, Ristanovic2011Mass}. A cost-effective scheme integrating both WiFi and cellular radio access technologies was proposed to efficiently address peak wireless data traffic and heterogeneous QoS requirements~\cite{Bennis2013CM}.
A subscribe-and-send architecture and an opportunistic forwarding protocol were presented in~\cite{Xiaofeng2013CC} such that the users having subscribed contents from the Content Service Provider (CSP) can obtain these contents from other users who can access these contents through WiFi opportunistic peer-to-peer communications rather than directly downloading the subscribed contents from the CSP.
In~\cite{Sou2013TVT}, the authors proposed an enhanced WiFi offloading model to bring mobile IP integration into the core network with Policy and Charging Control (PCC), and developed a comprehensive analytical model to quantify the performance of data offloading in terms of the amount of 3G resources saved by offloading and the deadline assurance for measuring the quality of user experience with PCC support.
In~\cite{Singh2013TWC}, the authors focused on the effect of inter-radio access technology (RAT) offloading on the overall system performance, and developed a general and tractable model that consisted of $M$ different RATs, each deploying up to $K$ different tiers of access points with different parameters. 
In contrast to these existing works, this paper is the first to investigate the economic behavior of WiFi offloading for two types of selfish players, which compete not only with the players of the same type but also with players of the other type. This sets our work apart from the existing literature in this field.

\subsection{Organization of Paper}

The remainder of this paper is organized as follows: 
The considered problem is formulated in Section~\ref{section:Problem_formulation}. In Section~\ref{sec:without_CB}, we analyzes the Stackelberg equilibrium without offloading capacity limit, while in Section~\ref{sec:with_CB} we analyze the Stackelberg equilibrium with offloading capacity limit. Simulation results are provided in Section~\ref{sec:sim}. Finally, the paper is concluded in Section~\ref{section:conclusion}.

\section{Problem Formulation}
\label{section:Problem_formulation}

In this section, we first provide the system model of WiFi offloading, and then introduce the pricing market framework. Subsequently, we formulate the problem to a Stackelberg game.

\subsection{System Model}

We consider a set $\mathcal{F}$ of mobile data flows (or data traffics) in a cellular network where each flow $f$ transmits a number of data packets from the source $S_f$ to the destination $D_f$. 
A set $\mathcal{R}$ of potential offloading APs (with $|\mathcal{R}|=R\ge2$) in the vicinity of the flows, may help flow $f$ to offload its data packets to the destination via another transmission network, e.g. WiFi. In return, the APs may obtain a certain reimbursement from flow $f$. 
The APs are assumed to be WiFis operating on different carriers, and accordingly the APs' signals do not mutually interfere with each other. Assume that time is slotted, and there is a network-wide slot synchronization. We focus on how the packets of flow $f$ should be priced such that the APs have an incentive to offload data packets of flow $f$.


\subsection{Pricing Framework}

For a selfish AP $i$ ($i\in \mathcal{R}$), to incentivize offloading, it must receive some reimbursement that is greater than its offloading cost. For this purpose, each flow $f$ offers a payment of $C^f$ to incentivize APs to offload data traffic, where $C^f$ is determined by the flow itself, i.e., $C^f$ is the strategy of flow $f$. We denote by $r^f_i$ the amount of data offloaded by AP $i$ for flow $f$. Hence, the utility of flow $f\in \mathcal{F}$ is defined as the net payoff that $f$ gets per slot:
\begin{equation}
\label{ob_flow}
    U_f\triangleq u_f\Big(\sum_{i\in \mathcal{R}} \log (1+r^f_i)\Big)-C^f,
\end{equation}
where the $\log (1+r^f_i)$ term\footnote{We adopt $\log (1+r^f_i)$ only for presentation purpose. This term can be replaced by other types of utility functions as long as they reflect the diminishing utility of flow $f$ in terms of $r^f_i$} reflects the diminishing utility of flow $f$ from $r^f_i$. Function $u_f(\cdot)$ represents the total utility from the assistance of all APs. We assume $u_f(w)$ is continuously differentiable, strictly increasing, and weakly concave in $w$, i.e., $u^{'}_f(w)> 0$ and $u^{''}_f(w)\leq0$, with $u_f(0)=0$.

Next, we consider the utility of the APs. For flow $f~(f\in \mathcal{F})$, the payment of $C^f$ is shared in accordance with the level of cooperation, i.e., the amount of data offloaded by the APs that offload packets of flow $f$. The vector $\mathbf{r}_i=\{ r^f_i,f\in \mathcal{F} \}$ is the strategy of AP $i$ where $\sum_{f\in \mathcal{F}} r^f_i \leq B$ reflects the limited offloading capacity $B$ of AP $i$. We denote the cost (e.g. in terms of energy) for AP $i$ to offload a packet of flow $f$ by $e^f_i$. Thus, the expected payoff per slot for AP $i$ is
\begin{equation}
\label{ob_node}
    V_i\triangleq \sum_{f\in \mathcal{F}}V^f_i= \sum_{f\in \mathcal{F}} \Big[ C^f \frac{r^f_i}{\sum_{j\in \mathcal{R}}r^f_j} - e^f_i r^f_i \Big],
\end{equation}
where $V^f_i \triangleq C^f \frac{r^f_i}{\sum_{j\in \mathcal{R}}r^f_j} - e^f_i r^f_i$.

The payoff function of AP $i$ has the following property.
\begin{lemma}
\label{lemma:concave}
$V_i$ is not a concave function in $r^f_j$ ($j\in \mathcal{R}, j\neq i$).
\end{lemma}
\begin{proof}
It is easy to show $\frac{ \partial^2 V_i}{\partial^2 r^f_j}\ge 0$, which means that $V_i$ is not a concave function in $r^f_j$ ($j\in \mathcal{R}, j\neq i$).
\end{proof}

\subsection{Stackelberg Game}
We model the offloading problem with pricing as a Stackelberg game which includes two roles (leader and follower) and two stages. In the first stage, each flow $f$ (as a leader) announces its reimbursement $C^f$, and the reimbursement from all flows are collected in a reimbursement vector $\mathbf{C}=(C^1,C^2,\cdots, C^{|\mathcal{F}|})$. In the second stage, each offloading AP $i$ (as a follower) in $\mathcal{R}$ choose its offloading size $\mathbf{r}_i=(r^1_i,r^2_i,\cdots,r^{|\mathcal{F}|}_i)$ for different flows to maximize its own utility. Hence, the flows are the leaders and the APs are the followers in this Stackelberg game. For convenience, let $\mathbf{r}=(\mathbf{r}_1,\mathbf{r}_2,\cdots,\mathbf{r}_{|\mathcal{R}|})$ denote the strategy profile of all APs where $\mathbf{r}_i$ is the strategy profile of AP $i$. Let $\mathbf{r}_{-i}$ denote the strategy profile excluding $\mathbf{r}_i$ and $\mathbf{r}^f_{-i}$ be the profile excluding AP $i$ given $f$. Then, $\mathbf{r}=(\mathbf{r}_i, \mathbf{r}_{-i})$ and $\mathbf{r}_i=(r^f_i,\mathbf{r}^f_{-i})$.

\subsubsection{Followers' Game}
Given $\mathbf{r}_{-i}$, each follower (AP $i$) chooses its strategy $\mathbf{r}_i$ to maximize its utility in response to the leaders' strategies $\mathbf{C}\triangleq(C^f, \mathbf{C}^{-f})=(C^1,C^2,\cdots, C^{|\mathcal{F}|})$. Thus, the objective of AP $i$ is to solve the following optimization problem:
 \begin{align}
 \label{game_f}
    & \widetilde{\mathbf{r}}_i(\mathbf{C})=\argmax_{\mathbf{r}_i} V_i(\mathbf{r}_i, \mathbf{r}_{-i},\mathbf{C}) \\
 \label{game_f_con1}
   s.t. & \sum_{f\in \mathcal{F}} r^f_i \leq B, ~~ \forall i\in \mathcal{R} \\
 \label{game_f_con2}
        & r^f_i\geq 0,~~  \forall i\in \mathcal{R},~ \forall f\in \mathcal{F}.
 \end{align}
Then, we have $ \widetilde{\mathbf{r}}(\mathbf{C})=\Big( \widetilde{\mathbf{r}}_1(\mathbf{C}),\cdots, \widetilde{\mathbf{r}}_{|\mathcal{R}|}(\mathbf{C})\Big)$.
Note that the followers' game itself can be considered as a non-cooperative game~\cite{Rosen1965ECO}.

\subsubsection{Leaders' game}
Given $\mathbf{C}^{-f}$, each leader (flow $f$) chooses its strategy $C^f$ to maximize its utility function $U_f(\cdot)$ anticipating that the followers will eventually respond with a collection of strategies that constitute an NE according to~\eqref{game_f}. Thus, the leaders' problem is
 \begin{align}
 \label{game_l}
    \widetilde{C}^f=\argmax_{C^f} U_f(C^f, \mathbf{C}^{-f},  \widetilde{\mathbf{r}}(C^f, \mathbf{C}^{-f})).
 \end{align}

The solution of the Stackelberg game is characterized by a Stackelberg Nash Equilibrium (SNE), that is a strategy profile from which no player has incentive to deviate unilaterally.

In the following sections, we will analyze the SNE for two different cases. In the first case, the capacity of the APs is not limited, which corresponds to omitting constraint~\eqref{game_f_con1}. In the second case, the capacity of the APs is limited, which corresponds to keeping constraint~\eqref{game_f_con1}.

\section{Stackelberg Game Equilibrium Analysis Without Capacity Bound}
\label{sec:without_CB}

In this section, we investigate the existence and uniqueness of an SNE for the considered Stackelberg game if the capacity of the APs is not limited (corresponding to omitting the constraint~\eqref{game_f_con1}). Specifically, we first show that the followers' game of the multi-leader multi-follower Stackelberg game (FG-MLMF-SG) can be decomposed into a series of followers' games. Then, we show the existence and uniqueness of an NE for the followers' game by analyzing its best response strategy, and prove the existence of a unique NE of the leaders' game by utilizing the structural properties of its objective function.

\subsection{Followers' Game}

Since the capacity of the APs is much larger than that of mobile devices, it is reasonable to assume that there is no offloading capacity limit for the APs. Under this assumption, the following proposition decomposes the complicated followers' game defined in Section~\ref{section:Problem_formulation} into a number of simpler games.

\begin{proposition}
\label{pr:sg_dec}
If the capacity of the APs is not limited, FG-MLMF-SG can be decomposed into $|\mathcal{F}|$ followers' games $\Big(FG(1), \cdots, FG(|\mathcal{F}|)\Big)$.
\end{proposition}
\begin{proof}
%
If the capacity of the APs is not limited, according to~\eqref{ob_node} and~\eqref{game_f}, FG-MLMF-SG, denoted by $\widetilde{\mathbf{r}}_i(\mathbf{C})=\argmax_{\mathbf{r}_i} V_i(\mathbf{r}_i, \mathbf{r}_{-i},\mathbf{C})$, can be decomposed into $|\mathcal{F}|$ followers games $\Big(FG(1), \cdots, FG(|\mathcal{F}|)\Big)$, where $FG(f), ~f \in \mathcal{F}$ corresponds to the optimization problem $\widetilde{r}^f_i(C^f)=\argmax_{r^f_i} V^f_i(r^f_i, \mathbf{r}^f_{-i},C^f)$.
\end{proof}

\begin{definition}
Given $C^f$ and $\mathbf{r}^f_{-i}$, a strategy is the best response strategy of AP $i$ for $FG(f)$, denoted by $\Gamma^f_i(\mathbf{r}^f_{-i})$, if it maximizes $V^f_i(r^f_i, \mathbf{r}^f_{-i})$ over $r^f_i\geq0$.
\end{definition}


From $ \frac{\partial V^f_i}{\partial r^f_i}=0$, we obtain $
    \widetilde{r}^f_i
    =\sqrt{\frac{C^f \sum_{j\in \mathcal{R}\setminus\{i\}} \widetilde{r}^f_j}{e^f_i}}-\sum_{j\in \mathcal{R}\setminus\{i\}} \widetilde{r}^f_j
$.
Therefore, the best response $\Gamma^f_i(\mathbf{r}^f_{-i})$ of follower $i$ for flow $f$ is
\begin{equation}
\label{node_br}
    \Gamma^f_i(\mathbf{r}^f_{-i})=
    \begin{cases}
        \sqrt{\frac{C^f \sum_{j\in \mathcal{R}\setminus\{i\}} \widetilde{r}^f_j}{e^f_i}}-\sum_{j\in \mathcal{R}\setminus\{i\}} \widetilde{r}^f_j, & \text{ if  } e^f_i\sum_{j\in \mathcal{R}\setminus\{i\}} \widetilde{r}^f_j \leq C^f\\
        0, & \text{ otherwise}.
    \end{cases}
\end{equation}
The best responses of follower $i$ for $\Big(FG(1), \cdots, FG(|\mathcal{F}|)\Big)$ are collected in the best response vector $\Gamma_i(\mathbf{r}_{-i})=\Big(\Gamma^1_i(\mathbf{r}^1_{-i}),\cdots,\Gamma^{|\mathcal{F}|}_i(\mathbf{r}^{|\mathcal{F}|}_{-i})\Big)$.
%

The following theorem states that the best response strategy leads to an NE of the FG-MLMF-SG.
\begin{theorem}
\label{th:exist_f_no_lim}
The strategy profile  $\mathbf{\widetilde{r}}=(\mathbf{\widetilde{r}}^1,\mathbf{\widetilde{r}}^2,\cdots,\mathbf{\widetilde{r}}^{|\mathcal{F}|})$ is an NE of the FG-MLMF-SG, where $\mathbf{\widetilde{r}}^f=(\widetilde{r}^f_1,\widetilde{r}^f_2,\cdots,\widetilde{r}^f_{|\mathcal{R}|})$ is an NE of $FG(f)$, where
\begin{enumerate}
\item the optimal sets of offloading APs, denoted by $\mathcal{S}=(\mathcal{S}_1,\mathcal{S}_2,\cdots,\mathcal{S}_F)$, are computed by Algorithm~\ref{ag:S_set};
  \item $\widetilde{r}^f_i=\frac{(|\mathcal{S}_f|-1)C^f}{\sum_{j\in \mathcal{S}_f}e^f_j}\Big(1-\frac{(|\mathcal{S}_f|-1)e^f_i}{\sum_{j\in \mathcal{S}_f}e^f_j}\Big)$ if $i\in \mathcal{S}_f$; $\widetilde{r}^f_i=0$ otherwise.
\end{enumerate}
\begin{algorithm}
\label{ag:S_set}
\caption{Computation of the optimal sets of offloading APs}
\begin{algorithmic}[1]
\FOR {$f \in \mathcal{F}$}
    \STATE Sort APs according to their offloading costs: $e^f_{\sigma_1} \leq e^f_{\sigma_2} \leq \cdots \leq e^f_{\sigma_R}$;
    \STATE $\mathcal{S}_f = \{\sigma_1,\sigma_2\}, i  = 3$;
    \WHILE {$i \leq R$ and $e^f_{\sigma_i} < \frac{\sum_{j\in \mathcal{S}_f}e^f_{j}}{|\mathcal{S}_f|-1}$}
        \STATE $\mathcal{S}_f=\mathcal{S}_f \cup \{\sigma_i\}, i=i+1$;
    \ENDWHILE
\ENDFOR
\RETURN {$\mathcal{S}=(\mathcal{S}_1,\mathcal{S}_2,\cdots,\mathcal{S}_F)$}.
\end{algorithmic}
\end{algorithm}

\end{theorem}

\begin{proof}
Please refer to Appendix~\ref{apdix:exist_f_no_lim}.
\end{proof}

After proving the existence of an NE of the FG-MLMF-SG, we next prove the uniqueness of the NE.


\begin{theorem}
\label{th:uniqu_f_no_lim}
Given $C^f$, denote the strategy profile of an NE by $\hat{\mathbf{r}} = (\hat{\mathbf{r}}^1,\hat{\mathbf{r}}^2,\cdots,\hat{\mathbf{r}}^{|\mathcal{F}|})$, where $\hat{\mathbf{r}}^f=(\hat{r}^f_1,\hat{r}^f_2,\cdots,\hat{r}^f_{|\mathcal{R}|})$, and define $\hat{\mathcal{S}}_f=\{i\in \mathcal{R}: \hat{r}^f_i > 0\}$. Then, we have
\begin{enumerate}
   \item $\hat{r}^f_i=\frac{(|\mathcal{\hat{S}}_f|-1)C^f}{\sum_{j\in \mathcal{\hat{S}}_f}e^f_j} \Big(1-\frac{(|\mathcal{\hat{S}}_f|-1)e^f_i}{\sum_{j\in \mathcal{\hat{S}}_f}e^f_j}\Big)$ if $i\in \hat{\mathcal{S}}_f$; $\hat{r}^f_i=0$ otherwise;
  \item We sort $\{e^f_j: j\in \mathcal{R} \}$ to $e^f_{\sigma_1} \leq e^f_{\sigma_2} \leq \cdots \leq e^f_{\sigma_R}$, then $\mathcal{\hat{S}}_f=\{\sigma_1,\cdots,\sigma_i\}$, where $\sigma_1,\cdots,\sigma_R$ is a permutation of $\mathcal{R}$ given $f$, $e^f_{\sigma_{i+1}} \geq \frac{\sum^{i}_{j=1} e^f_{\sigma_j} }{i-1}$, and $i\geq2$.


\end{enumerate}
These statements imply that the FG-MLMF-SG has a unique NE.
\end{theorem}

\begin{proof}
Please refer to Appendix~\ref{apdix:uniqu_f_no_lim}.
\end{proof}

Theorem~\ref{th:exist_f_no_lim} and Theorem~\ref{th:uniqu_f_no_lim} imply that there exists a unique NE in the FG-MLMF-SG.

\subsection{Leaders' Game}

According to the above analysis, the flows, which are the leaders in the MLMF-SG, know that there exists a unique NE for the APs for any given pricing vector $\mathbf{C}$. Hence, each flow $f$ can maximize its benefit by setting $C^f$.

Given a specific flow $f$, feeding back into~\eqref{ob_flow}, we have
\begin{equation*}
    U_f= u_f\Big(\sum_{i\in \mathcal{R}} \log (1+r^f_i)\Big)-C^f=u_f\Big(\sum_{i\in \mathcal{S}_f} \log \big(1+C^f k_i \big)\Big)-C^f,
\end{equation*}
where $k_i = \frac{|\mathcal{S}_f|-1}{\sum_{j\in \mathcal{S}_f}e^f_j} \Big(1-\frac{(|\mathcal{S}_f|-1)e^f_i}{\sum_{j\in \mathcal{S}_f}e^f_j}\Big)$.
\begin{theorem}
There exists a unique NE of the leaders' game in the MLMF-SG.
\end{theorem}
\begin{proof}
Given a specific flow $f$, the second derivative of $U_f$ with respect to $C^f$ is
\begin{align*}
    \frac{\partial^2 U_f}{\partial^2 C^f}= u''_f\Big(\sum_{i\in \mathcal{S}} \log (1+C^f k_i)\Big) \Big(\sum_{i\in \mathcal{S}}\frac{k_i}{1+C^f k_i}\Big)^2-u'_f\Big(\sum_{i\in \mathcal{S}} \log (1+C^f k_i)\Big)\sum_{i\in \mathcal{S}}\frac{k^2_i}{[1+C^f k_i]^2} < 0.
\end{align*}
Thus, $U_f= u_f\Big(\sum_{i\in \mathcal{R}} \log (1+r^f_i)\Big)-C^f$ is concave in $C^f$ for $C^f\in[0,\infty)$. Since $U_f|_{C^f=0}=0$ and $U_f|_{C^f=\infty}=-\infty$, $U_f$ has a unique maximizer, denoted by $\widetilde{C}^f=\argmax_{C^f}U^f$.
The $\widetilde{C}^f,~f\in \mathcal{F}$, compose the price vector $\widetilde{\mathbf{C}}$ which achieves the unique NE of the leaders' game in the MLMF-SG.
\end{proof}

Thus far, we have established the existence and uniqueness of the NE for the MLMF-SG when the offloading capacity of the APs is not limited. However, due to hardware limitation and energy consumption limits, in practice, constraints on the APs' offloading capability are inevitable, which makes the interaction between APs and flows more complex. In the next section, we will further study the properties of the NE of the MLMF-SG if a constraint on APs' offloading capacity is present.

\section{Stackelberg Game Equilibrium Analysis With Capacity Bound}
\label{sec:with_CB}

In the previous section, we have analyzed the NE of the considered MLMF-SG for the case when the offloading capacity of the APs is not limited. Now, we consider the game if a capacity constraint on the APs is present, and characterize the properties of the NE. 
First, we establish some structural properties of some relevant quantities in the leaders' game, and then we prove the existence and uniqueness of the NE for the leaders' game in the MLMF-SG. Finally, we present a distributed pricing algorithm for the leaders' game that converges to the unique equilibrium.

\subsection{Followers' Game}

To make the analysis of the game tractable, we assume that the offloading cost of a specific flow does not depend on the APs, that is, $e^f_i=e^f$ for any AP $ i\in \mathcal{R}$ given $f$. Note that this assumption is reasonable as all APs are assumed to be located in the vicinity of flow $f$.

We commence our discussion of the properties of the equilibrium by considering the best response of AP $i$ using the strategy $\mathbf{r}_i=(r^1_i,\cdots,r^{|\mathcal{F}|}_i)$. The corresponding optimization problem from the perspective of AP $i$ can be stated as:
\begin{align}
 \label{op_node}
\max_{\mathbf{r}_i} V_i(\mathbf{r}_i,\mathbf{r}_{-i})   ~~ s.t. ~~  \sum_{f\in \mathcal{F}} r^f_i \leq B,~~  r^f_i \geq 0,~~\forall f\in \mathcal{F}.
\end{align}

Thus, the corresponding Lagrangian function is given by:
\begin{align}
\label{op_node_lag}
    L(\mathbf{r}_i,\lambda_i,\mathbf{\nu})&=V_i(\mathbf{r}_i,\mathbf{r}_{-i})
    -\lambda_i\cdot\Big(\sum_{f\in \mathcal{F}} r^f_i - B\Big)+\sum_{f\in \mathcal{F}} \nu^f_i r^f_i.
\end{align}

Since $V_i$ is continuously differentiable in $r^f_i$, it follows that the Karush-Kuhn-Tucker (KKT) conditions corresponding to problem~\eqref{op_node_lag} are necessary for optimality. On the other hand, we note from~\eqref{ob_node} that, for a fixed $\mathbf{r}_{-i}$, function $V_i(\mathbf{r}_i,\mathbf{r}_{-i})$ is concave in $\mathbf{r}_i$ although it is not concave in $\mathbf{r}$ according to Lemma~\ref{lemma:concave}. This implies that the KKT conditions are sufficient for optimality as well. Thus, we conclude that a strategy profile is an equilibrium if and only if (i.i.f) there exist $\lambda_i \geq 0$ and $\{\nu^f_i \geq 0,f\in \mathcal{F}\}$ such that the following conditions are satisfied:
\begin{align*}
    (A_1):&  \quad \frac{\partial V_i}{\partial r^f_i} = \lambda_i - \nu^f_i, ~~~\forall f\in \mathcal{F} \\
    (A_2):&  \quad \lambda_i\cdot\Big(\sum_f r^f_i - B\Big)=0 \\
    (A_3):&  \quad  \nu^f_i r^f_i= 0, ~~~\forall f\in \mathcal{F}.
\end{align*}



For ease of further discussion, we introduce the concept of strictly interior equilibrium which is formally defined as follows:

\begin{definition}
We say that an equilibrium is a strictly interior equilibrium if the offloading size of any AP $i\in \mathcal{R}$ for any flow $f\in \mathcal{F}$ is strictly positive, i.e., $r^f_i >0 $.
\end{definition}

Now, we are ready to provide the following theorem, which guarantees the symmetry of a strictly interior equilibrium.

\begin{theorem}
\label{th:sym_interior_equi}
If a strictly interior equilibrium exists in the followers' game, then it is symmetrical, i.e., $r^f_i=r^f$ for any $i\in \mathcal{R}$.
\end{theorem}
\begin{proof}
Please refer to Appendix~\ref{apdix:sym_interior_equi}.
\end{proof}

Thus, in the following, we focus on symmetric strategy profiles, that is, all nodes use a symmetric strategy, i.e., $r^f_i=r^f$ for any $i\in \mathcal{R}$. To this end, we define the function
\begin{equation*}
    g^f(r^f)\triangleq \frac{\partial V_i}{\partial  r^f_i }\Big|_{r^f_j = r^f,~ \forall j\in \mathcal{R}}=C^f  \frac{R-1}{R^2 r^f} - e^f = C^f h^f(r^f)- e^f,
\end{equation*}
where $h^f(r^f)\triangleq \frac{R-1}{R^2 r^f}$.

Given a symmetric strategy profile, by Theorem~\ref{th:sym_interior_equi}, the KKT conditions for~\eqref{op_node_lag} can be refined to the existence of $\lambda_i \geq 0$ and $\{\nu^f_i = 0,f\in \mathcal{F}\}$ such that (A1)-(A3) are satisfied.

Now, we are ready to state the main result of this subsection.
\begin{theorem}
\label{th:sym_exist}
For any vector of flow price $\mathbf{C}$, there exists a unique set of $\{\rho^f, f\in \mathcal{F}\}$ such that the symmetric strategy profile $\{r^f_j = \rho^f, j\in \mathcal{R}\}$ is a Nash equilibrium. Furthermore, there exist $\lambda\geq0$ and $\{\nu^f = 0, f\in \mathcal{F} \}$, such that
\begin{align*}
    (B_1):&  \quad g^f(\rho^f) = \lambda  - \nu^f, ~~~\forall f\in \mathcal{F} \\
    (B_2):&  \quad \lambda\Big(\sum_{f\in \mathcal{F}} \rho^f - B\Big)=0 \\
    (B_3):&  \quad  \nu^f \rho^f= 0, ~~~\forall f\in \mathcal{F}.
\end{align*}
\end{theorem}
\begin{proof}
Please refer to Appendix~\ref{apdix:sym_exist}.
\end{proof}


Based on Theorem~\ref{th:sym_exist}, we obtain that the solution of the following convex optimization problem is the NE of the followers' game in the MLMF-SG.
\begin{align}
\label{op_node_convex}
   \max_{\rho^1\cdots \rho^{|\mathcal{F}|}} \sum_{f \in \mathcal{F}}  \Big(C^f  \frac{R-1}{R^2} \log(\rho^f) - e^f \rho^f\Big)
~~s.t. ~~ \sum_{f\in \mathcal{F}} \rho^f \leq B,
       ~~ \rho^f > 0~~\forall f\in \mathcal{F},
\end{align}
which can be easily solved by software packages, such as Matlab.

\subsection{Leaders' Game}

In this subsection, we study the effect of the payment rate $C^f$ of a specific flow $f \in \mathcal{F}$ on the followers' symmetric equilibrium when all other rates $\mathbf{C}^{-f}$ remain fixed. To streamline the discussion, we express the value of $\rho^f$ of the equilibrium corresponding to a given $C^f$ as a function $\rho^f=\Psi(C^f)$ (since we focus only on $\rho^f$ and are not interested in the strategy values for other flows). Also, we define the value of $\lambda$ that satisfies condition (B1)-(B3) in the equilibrium as a function $\lambda=\Lambda(C^f)$.

We begin by exploring these functions for extreme values of $C^f$.
Clearly, for $C^f=0$, the utility of any AP cooperating with flow $f$ is non-positive, implying $\rho^f=\Psi(C^f=0)=0$. However, from the KKT conditions (B1)-(B3), we know $\rho^f>0$, which implies $C^f>0$.
Thus, we assume that $\rho^f$ must be larger than a infinitesimal positive value, i.e., $\rho^f=0^+$. Define $C^f =\Psi^{-1}(\rho^f=0^+)\triangleq\underline{C}^f$ and $\underline{\lambda}=\Lambda(C^f=\underline{C}^f)$. $\Lambda(C^f)$ and $\Psi(C^f)$ have the following properties.


\begin{lemma}
\label{lemma:prop}
$\Lambda(C^f)$ and $\Psi(C^f)$ have the following properties:
\begin{enumerate}
  \item $\lambda=\Lambda(C^f)$ is continuous and non-decreasing in $C^f$;
  \item $\rho^f=\Psi(C^f)$ is continuous, and strictly increasing in $C^f \in (0,\infty)$;
  \item $\rho^f=\Psi(C^f)$ is concave in $C^f \in (0,\infty)$;
\end{enumerate}
\end{lemma}
\begin{proof}
Please refer to Appendix~\ref{apdix:prop}.
\end{proof}

\begin{lemma}
\label{lemma:u_f_sym_conv}
For a fixed $\mathbf{C}^{-f}$, the function $U_f(C^f, \mathbf{C}^{-f})$ is concave in $C^f$.
\end{lemma}
\begin{proof}
Assuming the followers respond with a symmetric equilibrium, the first order derivative of utility function $U_f$ with respect to $C^f$ is given by
\begin{align}
\label{dUf_dC}
    \frac{\partial U_f}{\partial C^f} = u'_f\Big(R\log (1+\rho^f)\Big) \frac{R}{1+\rho^f}\frac{\partial \rho_f}{\partial C^f} -1,
\end{align}
where $\rho^f=\Psi(C^f)$. Since $u_f(\cdot)$ is concave by assumption, $R\log (1+\rho^f)$ is increasing and concave in $\rho^f$, and $\rho^f$ is concave in $C^f$ by Lemma~\ref{lemma:prop}, it follows that $ \frac{\partial U_f}{\partial C^f} $ is non-decreasing in $C^f$, i.e., $U_f$ is indeed concave in $C^f$.
\end{proof}

\begin{lemma}
\label{lemma:br_bound}
The best-response function $\Upsilon^f(\mathbf{C}^{-f})$ of flow $f$ is bounded by $0\leq \Upsilon^f(\mathbf{C}^{-f})\leq u_f\Big(R \log (1+B)\Big)$.
\end{lemma}
\begin{proof}
Notice that $U_f =u_f\Big(R \log (1+\rho^f)\Big)-C^f$. Obviously, for the best response, the utility is nonnegative (utility 0 can always be obtained by $C^f=0$). Hence, $0\leq \Upsilon^f(\mathbf{C}^{-f})\leq \max_{\rho^f} u_f\Big(R \log (1+\rho^f)\Big)=u_f\Big(R \log (1+B)\Big)$.
\end{proof}

%

Due to the concavity of $U_f$ in $C^f$ (Lemma~\ref{lemma:u_f_sym_conv}), a unique solution is guaranteed; furthermore, we observe that if $u_f$ is continuously differentiable, the best response function is continuous as well.

\begin{theorem}
\label{th:ex_uq}
If the followers always respond with their symmetrical NE, then an equilibrium of the leaders' game, i.e., an SNE of the overall system, exists and is unique.
\end{theorem}
\begin{proof}
Please refer to Appendix~\ref{apdix:ex_uq}.
\end{proof}

Thus far, we have obtained the static characteristics of the leaders' game, i.e., the existence and uniqueness of the equilibrium. Next, we analyze the dynamic behavior of the leaders' game, i.e., how the game converges to the equilibrium from any initial strategy profile by best-response strategy updates. Before delving into the convergence analysis, we discuss the monotonicity of the best response function $\Upsilon(\mathbf{C}^{-f})$ of flow $f$.

\begin{lemma}
\label{lemma:br_mono}
The best response $\Upsilon(\mathbf{C}^{-f})$ of flow $f$ is monotonic and non-decreasing in $C^{f'}$ for any $f'\in \mathcal{F}\setminus\{f\}$.
\end{lemma}
\begin{proof}
Please refer to Appendix~\ref{apdix:br_mono}.
\end{proof}

Now, we are ready to state the following theorem which characterizes the dynamic behavior of the leaders' game.
\begin{theorem}
\label{th:conv}
Given some initial price vector $\mathbf{C}(0)$, if each flow $f$ responds according to Algorithm~\ref{ag:SGL_sym}, where $\Big(\rho^1(n),\cdots, \rho^{|\mathcal{F}|}(n)\Big)$ can be obtained by solving~\eqref{op_node_convex}, that is, flow $f \in \mathcal{F}$ updates its strategy as $C^f(n+1)=\Upsilon(\mathbf{C}^{-f}(n))$,
then $\lim_{n\rightarrow\infty}\mathbf{C}(n)=\mathbf{C}^*$, where $\mathbf{C}^*$ is the equilibrium of the leaders' game.
\end{theorem}

\begin{proof}
Please refer to Appendix~\ref{apdix:conv}.
\end{proof}

Distributed Algorithm~\ref{ag:SGL_sym} computes the price $C^f(n+1)$ of flow $f$ ($f\in \mathcal{F}$) at $n+1$,  where the price $C^f(n+1)$ of flow $f$ depends on $\Big(\rho^1(n),\cdots, \rho^{|\mathcal{F}|}(n)\Big)$ rather than the price of other flows, i.e., $C^{f'}(n)$ ($f'\neq f$).

\begin{algorithm}
\caption{Computing price for flow $f$}
\begin{algorithmic}[1]\label{ag:SGL_sym}
\STATE  \textbf{input:} $\rho^1(n),\cdots, \rho^{|\mathcal{F}|}(n)$;
        \IF { flow $f\in \mathcal{F}$ updates its strategy}
        \IF {$\rho^f(n) + \sum_{f'\ne f}\rho^{f'}(n) <B$}
            \STATE $C^f(n+1)= u'_f(R\log(1+\rho^f(n)))\frac{R\rho^f(n)}{1+\rho^f(n)}$;
        \ELSE
            \STATE $\lambda=\Big[\frac{u'_f(R\log(1+\rho^f(n)))}{\rho^f(n)(1+\rho^f(n))}\frac{R-1}{R}-\sum_{f'\in \mathcal{F}}\frac{e^{f'}}{\rho^{f'}(n)}\Big] \frac{1}{\sum_{f'\in \mathcal{F}}\frac{1}{\rho^{f'}(n)}}$;
            \STATE $C^f(n+1)=\rho^f(n)(\lambda+e^f)\frac{R^2}{R-1}$ for flow $f$;
        \ENDIF
    \ENDIF
\STATE \textbf{ouput:} {$C^f(n+1)$}.
\end{algorithmic}
\end{algorithm}

In this section, when the capacity of APs is limited, by considering a symmetric strategy profile, we have established the existence and uniqueness of the equilibrium of the corresponding MLMF-SG, and further, based on the best response strategy, presented a distributed price algorithm that allows the flows to computer their price independently.

\section{Numerical Simulation}
\label{sec:sim}

In this section, we demonstrate some of the theoretical results derived in this paper, and gain further insight into the behavior of the game for different scenarios via a numerical study. Our goal is to present several scenarios indicative of the typical interactions among the players in the game. First, we consider the case when the offloading capacity of the APs is not limited. Specifically, we evaluate the effect of the offloading cost, heterogeneity of traffics, the number of APs etc, on the performance of the equilibrium. In the second part of this section, we evaluate the performance of the game when the offloading capacity of APs is limited.

\subsection{Multiple Cellular Flows and Multiple APs with Offloading Capacity Limit}

First, we introduce the price of anarchy (PoA). Denote the unique equilibrium of the proposed MLMF-SG as $(\mathbf{r}^*_{ne}, \mathbf{C}^*_{ne})$, we know that $(\mathbf{r}^*_{ne}, \mathbf{C}^*_{ne})$ can be obtained by solving problem~\eqref{op_node_convex} and running Algorithm~\ref{ag:SGL_sym}, and the optimum system utility $U_{NE}$ at equilibrium is a function of $(\mathbf{r}^*_{ne}, \mathbf{C}^*_{ne})$, i.e., $U_{NE}=\Big[\sum_{f\in \mathcal{F}} U_f  + \sum_{i\in \mathcal{R}} V_i\Big]_{(\mathbf{r}, \mathbf{C})=(\mathbf{r}^*_{ne}, \mathbf{C}^*_{ne})} $. On the other hand, the social utility $U_{Opt}$ can be obtained by solving the following optimization problem,
\begin{align*}
   \max_{\rho^1\cdots \rho^{|\mathcal{F}|}} \Big\{U_s \triangleq\sum_{f\in \mathcal{F}} u_f\Big(R \log (1+\rho^f)\Big) - \sum_{f\in \mathcal{F}} R  e^f \rho^f \Big\} ~~s.t. ~~ \sum_{f\in \mathcal{F}} \rho^f \leq B, ~ \rho^f > 0, ~ f\in \mathcal{F}.
\end{align*}
Denote $\mathbf{\rho}^*=\argmax_{\rho^1\cdots \rho^{|\mathcal{F}|}}\{U_s\}$, and then, $U_{Opt}=U_s\big|_{\mathbf{\rho}^*}$. Therefore, PoA$=\frac{U_{Opt}}{U_{NE}}$.


\subsubsection{Convergence}
We first consider the simplest scenario with two cellular traffic flows $|\mathcal{F}|=2$  and two APs $|\mathcal{R}|=2$, which allows us to illustrate the interactions between flows and APs. Specifically, for the cellular traffic flow $f\in \mathcal{F}$, we adopt a linear utility function $U_f=\omega_f \sum_{i\in \mathcal{R}} \log (1+r^f_i)-C^f$. The parameters are set as follows: offloading costs $e^1=0.1$ and $e^2=0.3$, weight coefficients $w_1=1$ and $w_2=2$, and capacity limit $B=7$ in Fig.~\ref{fig:convergence}(a) and $B=1$ in Figs.~\ref{fig:convergence}(b)--(d), respectively. By solving problem~\eqref{op_node_convex}, 
we obtain $\rho^1=4$ and $\rho^2=2.33$, and further, $C^1=1.6$ and $C^2=2.8$ from $C^f= e^f\rho^f\frac{ R^2}{R-1}$ according to Algorithm~\ref{ag:SGL_sym}. Note that $\rho^1+\rho^2= 6.33<7$ implying that the condition $\rho^1+\rho^2<B$ holds, which is shown in Fig.~\ref{fig:convergence}(a).
On the other hand, for $B=1$, $\rho^1+\rho^2=1$ must be satisfied at the NE, which is illustrated in Figs.~\ref{fig:convergence}(b)--(d). Moreover, we observe from Figs.~\ref{fig:convergence}(b)--(d) that the price vector and the strategy profile converge from different initial price vectors $\mathbf{C}(0)=(0.01,0.01)$, $\mathbf{C}(0)=(5, 0.01)$, and $\mathbf{C}(0)=(10,10)$, respectively, which validates the proposed Algorithm~\ref{ag:SGL_sym}.

\begin{figure}[ht]
\begin{center}
\begin{tabular}{cc}
\includegraphics[scale=0.4]{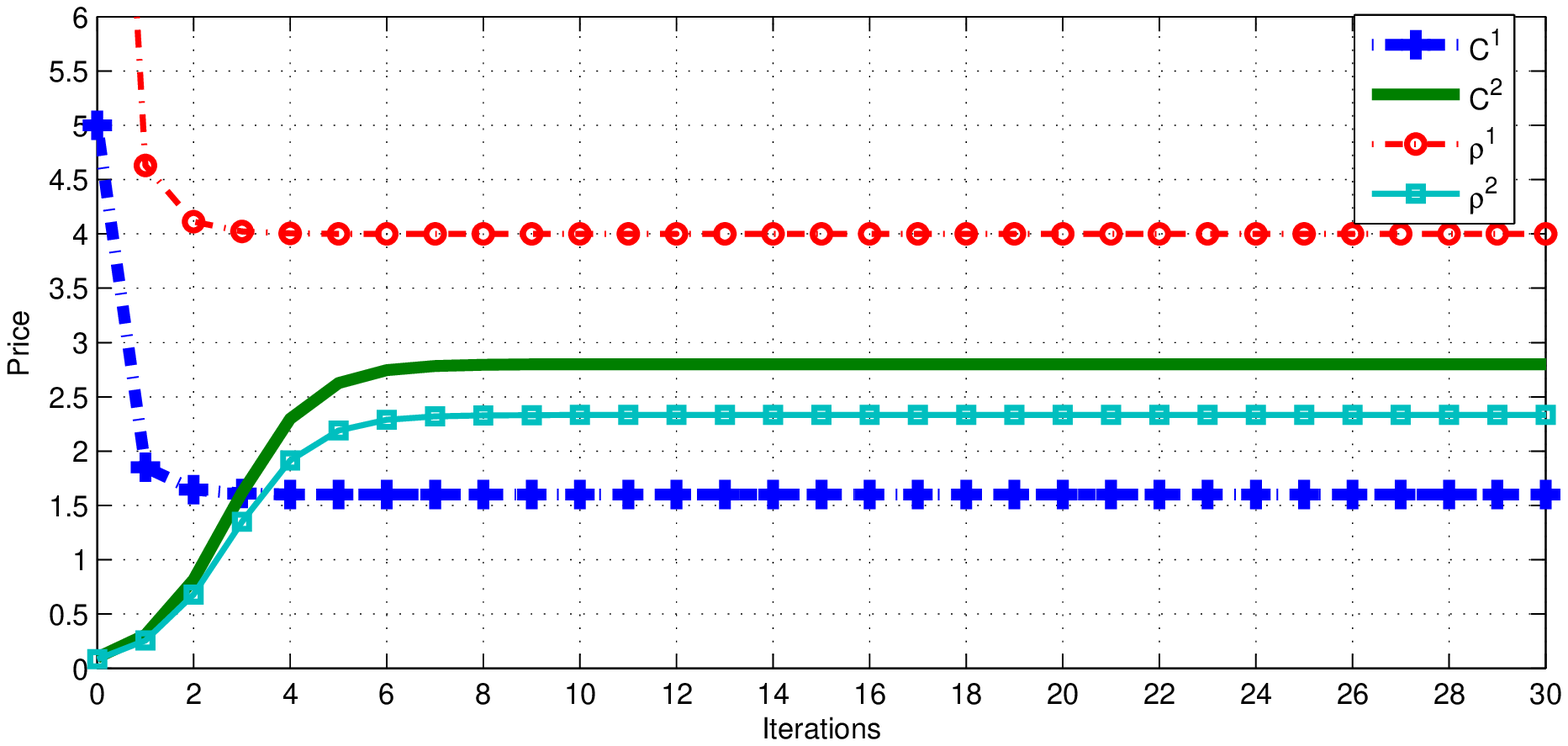} & \includegraphics[scale=0.4]{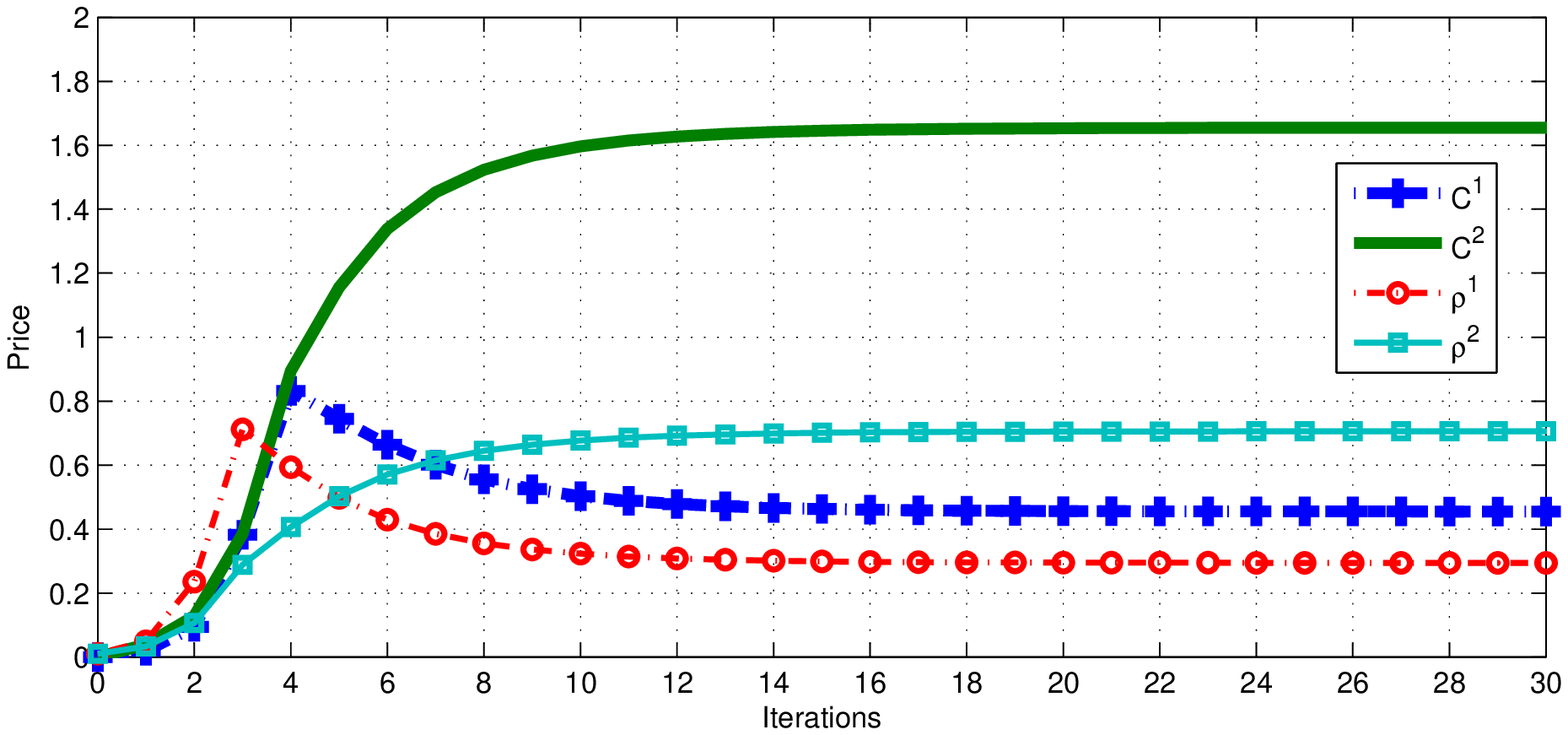}\\
(a) & (b)\\
\includegraphics[scale=0.4]{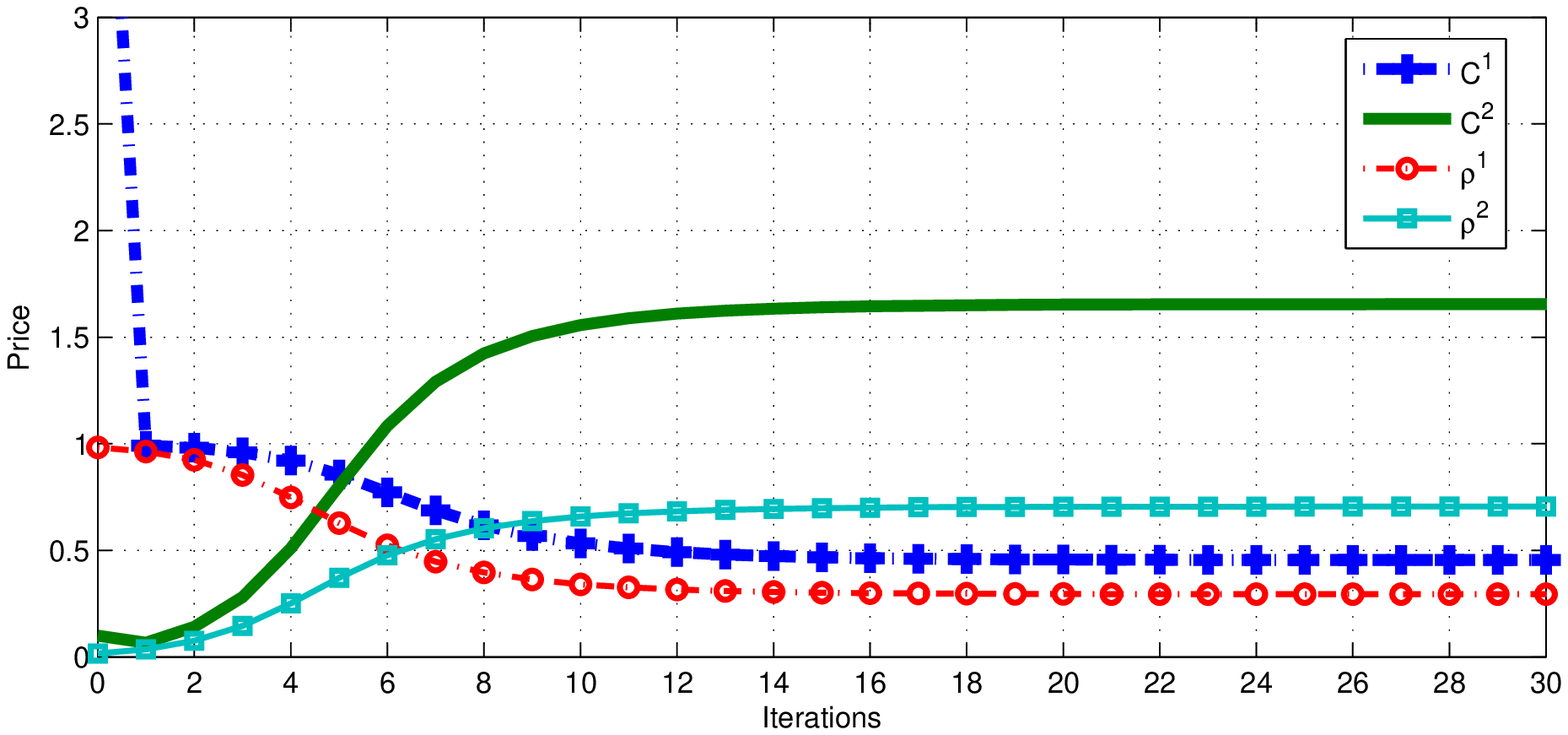} & \includegraphics[scale=0.4]{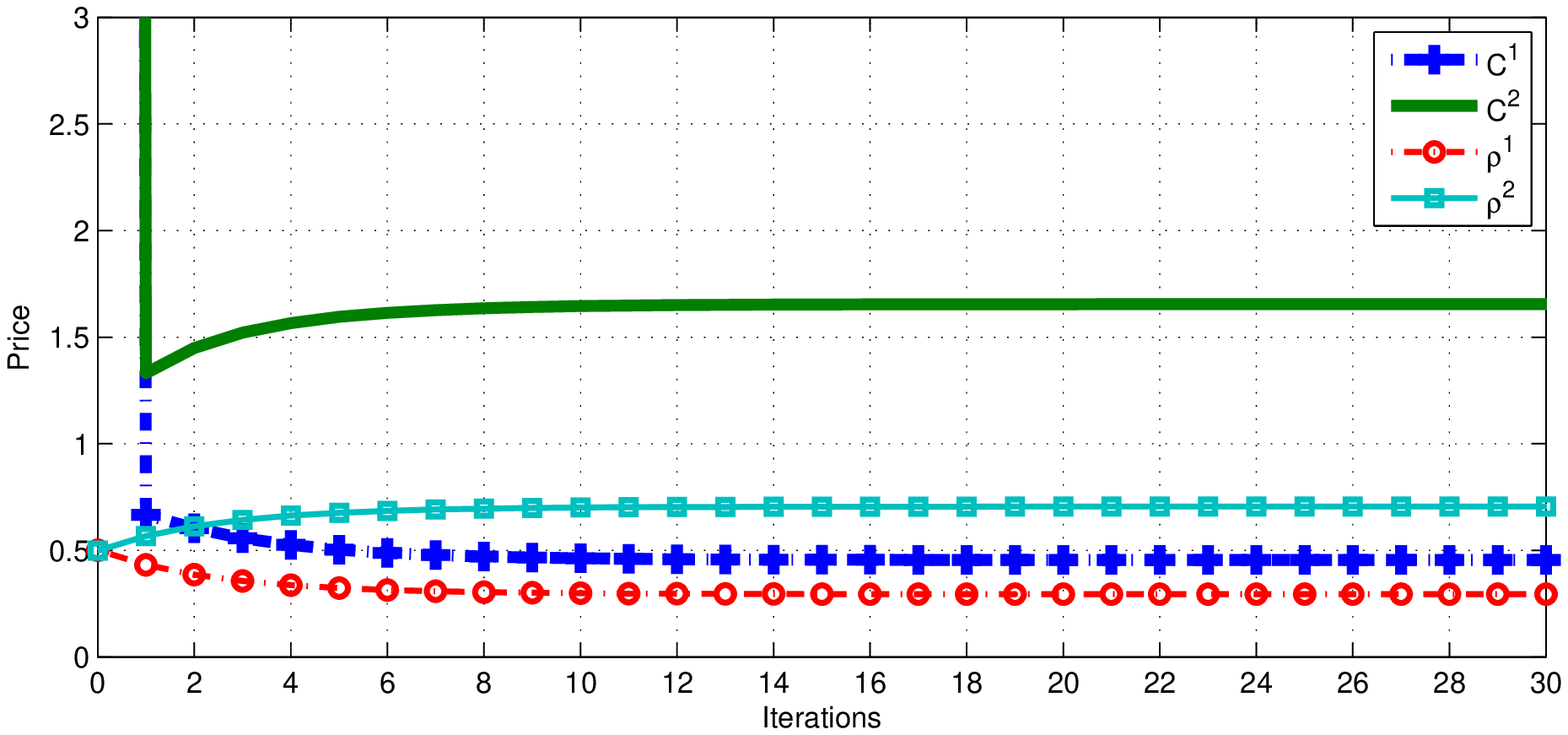}\\
(c) & (d)\\
\end{tabular}
\end{center}
\caption{Convergence of price vector and policy vector to NE: $e^1=0.1$, $e^2=0.3$, $w_1=1$, $w_2=2$, and $B=7$ in (a), and $B=1$ in (b)--(d). For Figs.~\ref{fig:convergence} (b)--(d), different initial price vectors $\mathbf{C}(0)=(0.01,0.01)$, $\mathbf{C}(0)=(5, 0.01)$, and $\mathbf{C}(0)=(10,10)$, are used, respectively. }
\label{fig:convergence}
\end{figure}
%
%
%
%

\subsubsection{Offloading Cost}
Considering two cellular traffics $|\mathcal{F}|=2$, two APs $|\mathcal{R}|=2$ and two linear utility functions with $w_1=2$ and $w_2=1$, respectively, we fix the offloading capacity to $B=2$ and the offloading cost of AP 1 to $e^1=0.5$, and then show how the price vector and strategy profile change as function of the offloading cost $e^2$ of AP 2. From Fig.~\ref{fig:flowcost2}, we make the following observations.

(a) $e^2\in(0.1,0.25]$. Assume $\rho^1 + \rho^2 <B$, then according to~\eqref{contrator_1}, we have $\rho^f=\frac{w_f}{e^f}\frac{R-1}{R}-1$ (which is decreasing in $e^f$), and furthermore, $\rho^1|_{e^1=0.5} =1$ and $\rho^2|_{e^2=0.1}=4$, which implies that $\rho^1|_{e^1=0.5} + \rho^2|_{e^2=0.1}=5>B=2$ contradicts the assumption $\rho^1 + \rho^2 <B$. Thus, $\rho^1+\rho^2=B=2$ must be satisfied for $e^2\in(0.1,0.25]$, which is shown in Fig.~\ref{fig:flowcost2};

(b) $e^2\in(0.25,0.5]$. The condition $\rho^1+\rho^2\leq B=2$ is met with equality until $e^2=0.25$. That is, when $e^2=0.25$, $\rho^2|_{e^2=0.25}=1$ and $\rho^1|_{e^1=0.5} + \rho^2|_{e^2=0.25}=2=B$. Hence, if $e^2>0.25$, then the condition $\rho^1+\rho^2 \leq B=2$ is no longer met with equality and $\rho^1+\rho^2<B$, which leads to $\rho^1=1$ and $C^1=w_1 R \frac{\rho^1}{1+\rho^1}=2$ from Algorithm~\ref{ag:SGL_sym}. Also, as $e^2$ increases, $\rho^2=\frac{w_2}{e^2}\frac{R-1}{R}-1$ decreases, which can be observed in Fig.~\ref{fig:flowcost2};

(c) $e^2\in(0.5,1]$. When $e^2\geq0.5$, $\rho^2=\frac{w_2}{e^2}\frac{R-1}{R}-1=\frac{1}{2 e^2}-1\leq0$, which implies that AP 2 does not offload any data and accordingly, $\rho^2 = 0$ and $C^2 = 0$.
\begin{figure}
\centering
\includegraphics[width=10cm, height=4cm]{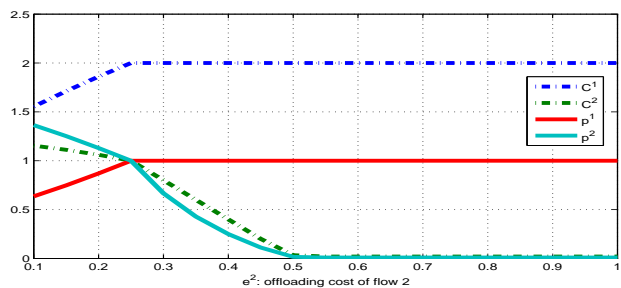}
\caption{The impact of offloading cost on the price vector and strategy profile. $w_1=2$, $w_2=1$, $B=2$, and $e^1=0.5$.}
\label{fig:flowcost2}
\end{figure}

\subsubsection{Heterogeneity of Data Traffic Flows}
For $|\mathcal{F}|=2$, $|\mathcal{R}|=2$, $w_1=1$, $e^1=0.1$, $e^2=0.3$, and $B=1$, Fig.~\ref{fig:PoA_w2} illustrates the relation between $w_2$ and the price of anarchy (PoA), defined as the ratio between the optimal social utility $U_{Opt}$ and the system utility $U_{NE}$ achieved at NE. In particular, the peak at $w_2=1$ can be explained as follows. When $w_1=w_2$, the two traffic flows are homogeneous. In this case, the APs prefer to offload the traffic flow $f=1$ because of its lower offloading cost (i.e., $e^1=0.1<0.3=e^2$), and flow $f=2$ cannot be treated equally, which leads to the peak of PoA at $w_1=w_2$. On the other hand, as the difference between $w_1$ and $w_2$ increases, corresponding to a larger heterogeneity of traffic flows, the NE strategy requires that the APs participate in pricing fully and offload each traffic flow as much as possible, and thus, PoA $\rightarrow1$.
\begin{figure}
\centering
\includegraphics[scale=0.5]{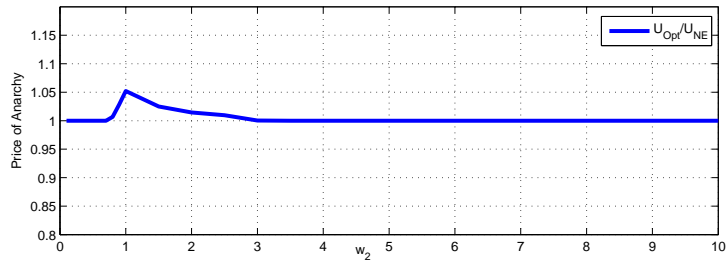}
\caption{The impact of the heterogeneity of traffic flows on PoA. $|\mathcal{F}|=2$, $|\mathcal{R}|=2$, $w_1=1$, $e^1=0.1$, $e^2=0.3$, and $B=1$.}
\label{fig:PoA_w2}
\end{figure}

\subsubsection{Number of APs}
We consider two data traffic flows and multiple APs with a linear utility function for each traffic flow, i.e., $u_f(\sum_{i\in \mathcal{R}} \log (1+r^f_i))= w_f \sum_{i\in \mathcal{R}} \log (1+r^f_i)$ where $w_1=2$ and $w_2=3$. Fig.~\ref{fig:MulRelayeps} reveals that the PoA decreases with the number of APs. Since the APs have the same offloading cost for each traffic flow, they equally and fully participate in the pricing process of each traffic flow, and thus the equilibrium utility $U_{NE}$ becomes much closer to the social utility $U_{Opt}$ as the increasing number of APs.
\begin{figure}
\centering
\includegraphics[scale=0.5]{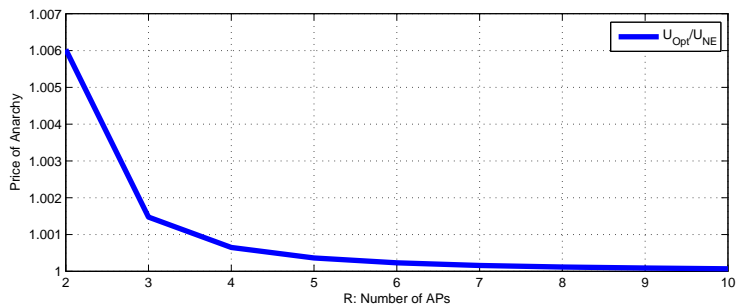}
\caption{The impact of the number of APs on PoA. $w_1=2$, $w_2=3$,  $e^1_i =0.1$. and $e^2_i =0.2$ for any $i \in \mathcal{R}$.}
\label{fig:MulRelayeps}
\end{figure}

\subsubsection{Large System}

To demonstrate the asymptotic properties of the game in a large-scale symmetric scenario, we consider 10 APs and 3 cellular traffic flows with $e^1=0.1$, $e^2=0.3$, $e^3=0.2$, and $w_f=1$, $f\in \mathcal{F}$, for more complex utility functions, namely, power-law functions $u_f(x)=w_f x^b,~0<b<1$, and the logarithmic function $u_f(x)=w_f \log(1+x)$. Fig.~\ref{fig:Price_diffunc} shows the PoA for the considered utility functions when the number of APs increases from $2$ to $10$. Specifically, for the logarithmic utility function as well as the linear function, the PoA decreases with increasing number of APs, while the PoA increases for the power-law utility functions. Moreover, the PoA increases as the  value of $b$ decreases (actually the linear function can be seen as a power-law function with $b=1$). Fig.~\ref{fig:Price_diffunc} suggests that the proposed framework can achieve an efficient equilibrium with at most $15\%$ loss of system utility compared to the optimum system utility when $3\leq R\leq 10$.
\begin{figure}
\centering
\includegraphics[scale=0.5]{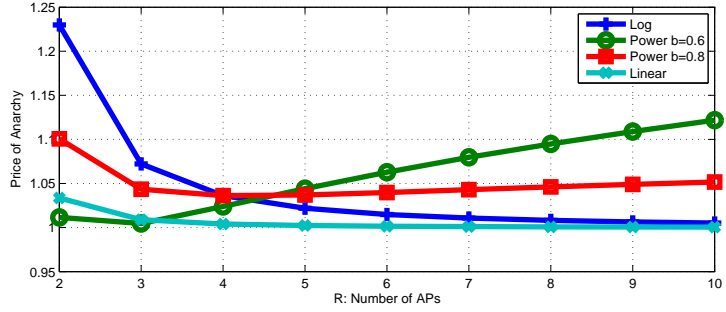}
\caption{Price of anarchy for log, power-law, and linear functions. $e^1=0.1$, $e^2=0.3$, $e^3=0.2$, and $w_1=w_2=w_3=1$.}
\label{fig:Price_diffunc}
\end{figure}

\subsection{Multiple Cellular Flows and Multiple APs without Offloading Capacity Limit}

In this subsection, for the case where the offloading capacity of the APs is not limited, we analyze how the PoA is affected by different parameters, i.e., the offloading cost, heterogeneity of traffic, the number of APs, and the number of traffic flows. First, we note that the optimum system utility, $U_{NE}$, at equilibrium is a function of $\mathbf{C}^*$ which is determined by the leaders' utility functions. Specifically, from $\frac{\partial U_f}{\partial C^f} = u'_f\Big(\sum_{i\in \mathcal{R}} \log (1+k^f_i C^f)\Big) \sum_{i\in \mathcal{R}} \frac{k^f_i}{1+k^f_i C^f}-1=0,~f\in\mathcal{F}$, we can obtain $\mathbf{C}^*$, and further, $U_{NE}=\Big[\sum_{f\in \mathcal{F}} u_f\Big(\sum_{i\in \mathcal{R}} \log (1+k^f_i C^f)\Big) - \sum_{f\in \mathcal{F}} \sum_{i\in \mathcal{R}}  e^f_i k^f_i C^f\Big]_{\mathbf{C}=\mathbf{C}^*} $ which shows that $U_{NE}$ is determined by the leaders' price vector. On the other hand, the social utility $U_{Opt}$ is the maximum of $U_s =\sum_{f\in \mathcal{F}} u_f\Big(\sum_{i\in \mathcal{R}} \log (1+r^f_i)\Big) - \sum_{f\in \mathcal{F}}\sum_{i\in \mathcal{R}}  e^f_i r^f_i$. Hence, from $\frac{\partial U_s}{\partial r^f_i} = u'_f\Big(\sum_{i\in \mathcal{R}} \log (1+r^f_i)\Big) \sum_{i\in \mathcal{R}} \frac{1}{1+r^f_i} - e^f_i =0,~i\in \mathcal{R},~f\in \mathcal{F} $, we have the optimum $\mathbf{r}^*$ and further $U_{Opt}=U_s\big|_{\mathbf{r}^*}$ which shows that $U_{Opt}$ is determined by the offloading size of the followers. 

\subsubsection{Offloading Cost and Heterogeneity of Data Traffic Flows}
In this scenario, we consider two symmetric APs and two traffic flows. 
In particular, the offloading cost of the APs for flow $f$ is homogeneous, i.e., $e^1_1=e^1_2=0.2$ and $e^2_1=e^2_2=e^2$. Meanwhile, the utility function of each flow is assumed to be a linear function, i.e, $u_f(\sum_{i\in \mathcal{R}} \log (1+r^f_i))= w_f \sum_{i\in \mathcal{R}} \log (1+r^f_i)$ where $w_1=1$.
Fig.~\ref{fig:AsymW2e2} show how the PoA is affected by the offloading cost and the heterogeneity of flows. We observe that as $w_2$ increases, corresponding to an increasing heterogeneity of flows, the PoA tends to decrease and approaches 1; on the other hand, as $e^2$ increases from $0.2$ to $0.8$, PoA tends to increase. For example, when $w_2=2$, the APs are more reluctant to offload flow $f=2$ for its larger offloading cost, and accordingly, the two traffic flows are not treated equally. In this case, flow $f=2$ cannot participate in the market pricing to the same extent as its counterpart $f=1$, which leads to an increase of the PoA.
\begin{figure}
\centering
\includegraphics[width=10cm,height=4cm]{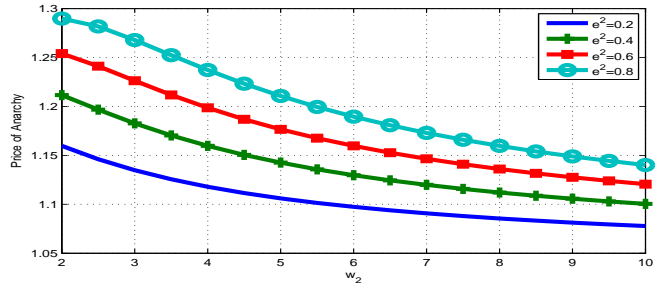}
\caption{The impact of the offloading cost and heterogeneity of traffic flows on PoA. $w_1=1$, $e^1_1=e^1_2=0.2$, and $e^2_1=e^2_2=e^2$.}
\label{fig:AsymW2e2}
\end{figure}

\subsubsection{Number of APs}
We consider two data traffic flows and multiple APs with a linear utility function for each traffic flow, i.e., $u_f(\sum_{i\in \mathcal{R}} \log (1+r^f_i))= w_f \sum_{i\in \mathcal{R}} \log (1+r^f_i)$ where $w_1=2$ and $w_2=3$.
We consider two kinds of APs: homogeneous APs (with the same offloading cost for each traffic flow) and heterogeneous APs (with different offloading costs for different traffic flows).

(a) Homogenous APs: Assume $e^1_i =0.1$ and $e^2_i =0.2$ for any $i \in \mathcal{R}$. From Fig.~\ref{fig:AsymRinc}, we observe that the PoA decreases as the number of the APs increases, and approaches 1 for $R \geq 5$. This can be explained as follows. As each traffic flow has the same offloading cost, the APs equally and fully participate in the pricing process, and as a consequence, the equilibrium utility $U_{NE}$ approaches the social utility $U_{Opt}$ as the number of APs increases.

(b) Heterogenous APs: In this case, we assume $e^1_1 =0.1, e^2_1 =0.2$, and $e^f_{j+1}=e^f_j+0.1$ to generate the offloading cost for each AP $j \in \mathcal{R}$, which reflects the different QoS requirements of the APs. From Fig.~\ref{fig:AsymRinc}, we observe that PoA tends to increase as the number of APs increases, and ultimately converges to a stable value. Because of the adopted generating rule for offloading cost, the offloading cost increases as the number of the APs. Hence, when the number of APs exceeds a certain threshold, the APs with larger offloading cost cannot obtain positive utility by offloading data traffic flow, and thus, do not offload data because of their selfishness. This is the reason for the stability of the PoA when the number of APs exceeds a certain threshold.
\begin{figure}
\centering
\includegraphics[scale=0.5]{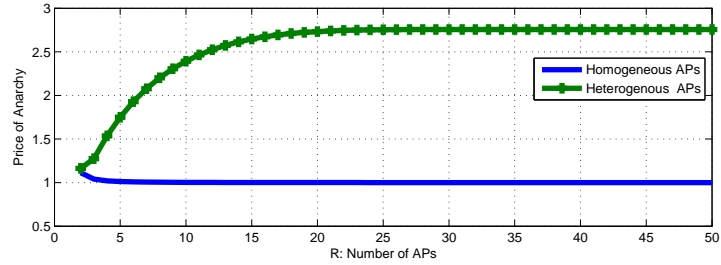}
\caption{The impact of the number of APs on PoA. (Homogeneous) $w_1=2$, $w_2=3$,  $e^1_i =0.1$ and $e^2_i =0.2$ for any $i \in \mathcal{R}$. (Heterogeneous)  $w_1=2$, $w_2=3$, $e^1_1 =0.1$, $e^2_1 =0.2$, and $e^f_{j+1}=e^f_j+0.1$ $j \in \mathcal{R}$.}
\label{fig:AsymRinc}
\end{figure}

\subsubsection{Number of Data Traffic Flows}
We consider two APs and multiple data traffic flows with a linear utility function for each traffic flow, i.e., $u_f(\sum_{i\in \mathcal{R}} \log (1+r^f_i))= w_f \sum_{i\in \mathcal{R}} \log (1+r^f_i)$. Specifically, we consider homogeneous and heterogeneous data traffic flows, respectively.

(a) Homogeneous Data Traffic Flows: We use the rule, $w_1=2$, $w_2=3$, and $w_{f+1}=w_{f}, 2\leq f\leq F$, to generate a linear utility for each flow, and set the offloading cost to $e^1_1 =0.1, e^1_2 =0.2$ and $e^f_1 =0.3, e^f_2 =0.4$ for $2\leq f\leq F$. From Fig.~\ref{fig:AsymFinc}, we observe that the PoA decreases with the number of flows. This can be explained as follows. As the number of homogenous flows increases, a larger number of flows compete in the market, which makes the $U_{NE}$ more efficient and close to the optimum social utility.

(b) Heterogeneous Data Traffic Flows: We use the rule, $w_1=2$, $e^1_1 =0.1$, and $e^1_2 =0.2$, $w_{f+1}=w_f+1$, and $e^{f+1}_j=e^f_j+0.1$ to generate a linear utility for each flow and a offloading cost for each AP, respectively. Fig.~\ref{fig:AsymFinc} shows that the PoA first decreases steeply and then increases slowly as the number of flows increases. In the steep region ($F\leq5$), the effect of the heterogeneous utility functions dominates the effect of the heterogenous offloading costs, which incentivizes the APs to offload data traffic flows, and as a consequence, $U_{NE}$ comes closer to $U_{Opt}$. In the flat region ($F>5$), as the number of flows increases, the advantage of the heterogeneous utility functions decreases while the negative effect of the offloading cost becomes much stronger, i.e., the APs have less incentive to offload data traffic flows, and consequently, the PoA begins to increase slowly.
\begin{figure}
\centering
\includegraphics[scale=0.6]{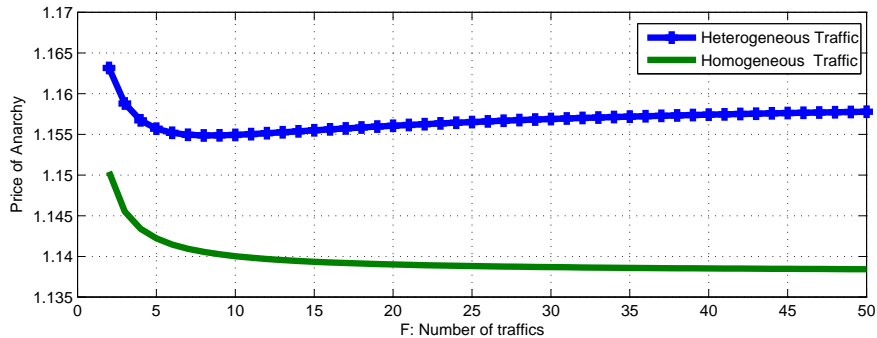}
\caption{The impact of the number of traffics on PoA. (Homogeneous) $w_1=2$, $w_2=3$, and $w_{f+1}=w_{f}, 2\leq f\leq F$. (Heterogeneous) $w_1=2$, $e^1_1 =0.1$, and $e^1_2 =0.2$, $w_{f+1}=w_f+1$, and $e^{f+1}_j=e^f_j+0.1$.}
\label{fig:AsymFinc}
\end{figure}

\section{Conclusions}\label{section:conclusion}

In this paper, we have proposed a pricing framework for cellular networks to offload mobile data traffic with the assistance of WiFi network. Specifically, the proposed framework can be utilized to motivate offloading service providers to participate in mobile data offloading, which is a new paradigm to alleviate cellular network congestion and to improve the level of user satisfaction as well. We have modeled the pricing mechanism as a multi-leader multi-follower Stackelberg game in which the offloading service providers are the followers and the offloading service consumers are the leaders. Technically speaking, we have analyzed the proposed Stackelberg game by distinguishing two different cases based on the offloading capacity of the APs.
For the case where the APs do not have an offloading capacity limit, we have decomposed the followers' game of the multi-leader multi-follower Stackelberg game into a fixed number of followers' games, and proved the existence and uniqueness of the equilibrium, and obtained an efficient algorithm to compute the equilibrium.
For the case with offloading capacity limit, by considering the symmetric strategy profile, we have established some structural results for the equilibrium, and further proved the existence and uniqueness of the equilibrium of the Stackelberg game. Consequently, we presented a distributed algorithm to compute the offloading price for each flow, and proved its convergence to the unique equilibrium. Finally, extensive numerical experiments were provided to demonstrate that the Stackelberg equilibrium is very close to the corresponding social optimum for both considered cases.

There are some future research directions. One direction is the investigation of the case of asymmetric strategy profiles, where each AP may have different offloading cost and offloading capacity. One of the possible approaches is to analyze the followers' game as a non-cooperative game~\cite{Rosen1965ECO}, and we expect that a unique equilibrium exists. Another promising direction for research is the investigation of other allocation rules for reimbursement (rather than the proportional allocation rule considered in this paper) such that the offloading service providers can be motivated to participate in mobile data offloading more actively.


\newpage

\appendices
\section{Proof of Theorem~\ref{th:exist_f_no_lim}}
\label{apdix:exist_f_no_lim}
Based on Proposition~\ref{pr:sg_dec}, to prove $\mathbf{\widetilde{r}}=(\mathbf{\widetilde{r}}^1,\mathbf{\widetilde{r}}^2,\cdots,\mathbf{\widetilde{r}}^{|\mathcal{F}|})$ is an NE of the FG-MLMF-SG, we only need to show that for a given flow $f$ the strategy profile $\mathbf{\widetilde{r}}^f$ is an NE of the $FG(f)$.
%

First, we prove that for $FG(f)$ and any $i \notin \mathcal{S}_f$, $\widetilde{r}^f_i=0$ is the best response strategy given $\widetilde{\mathbf{r}}^f_{-i}$. From Algorithm~\ref{ag:S_set}, we have $e^f_{i} \geq \frac{\sum_{j\in \mathcal{S}_f}e^f_{j}}{|\mathcal{S}_f|-1}$ for any $i\notin\mathcal{S}_f$ at the NE point. Since $i \notin \mathcal{S}_f$, we have $e^f_i \sum_{j\in \mathcal{S}_f\setminus\{i\}} \widetilde{r}^f_j  = e^f_i \sum_{j\in \mathcal{S}_f} \widetilde{r}^f_j=\frac{e^f_i(|\mathcal{S}_f|-1)}{\sum_{j\in \mathcal{S}_f}e^f_j}C^f\geq C^f$, which implies that $\Gamma^f_i(\mathbf{r}^f_{-i})=0$ according to~\eqref{node_br}.

Next, we prove that for $FG(f)$ and any $\sigma_i \in \mathcal{S}_f$, $\widetilde{r}^f_{\sigma_i}$ is the best response strategy given $\widetilde{r}^f_{-{\sigma_i}}$. Since $e^f_{\sigma_i} < \frac{\sum^i_{j=1}e^f_{\sigma_j}}{i-1}$ (line 4 of Algorithm~\ref{ag:S_set}), we have
\begin{align*}
    (|\mathcal{S}_f|-1)e^f_{\sigma_i} =(i-1)e^f_{\sigma_i}  + (|\mathcal{S}_f|-i)e^f_{\sigma_i} < \sum^i_{j=1}e^f_{\sigma_j} + \sum^{|\mathcal{S}_f|}_{j=i+1}e^f_{\sigma_j}=\sum_{j\in \mathcal{S}_f}e^f_j,
\end{align*}
where $e^f_{\sigma_i}\leq e^f_{\sigma_j}$ for $i+1\leq j \leq |\mathcal{S}_f|$.

Furthermore,
\begin{align*}
  e^f_{\sigma_i}  \sum_{j\in \mathcal{R}\setminus\{\sigma_i\}} \widetilde{r}^f_j = e^f_{\sigma_i}  \sum_{j\in \mathcal{S}_f\setminus\{\sigma_i\}} \widetilde{r}^f_j = \frac{(|\mathcal{S}_f|-1)^2( e^f_{\sigma_i} )^2}{[\sum_{j\in \mathcal{S}_f}e^f_j]^2} C^f<  C^f.
\end{align*}
According to~\eqref{node_br}, we have
\begin{align*}
    \Gamma^f_i(\mathbf{r}^f_{-{\sigma_i}})
    =\sqrt{\frac{C^f \sum_{j\in \mathcal{R}\setminus\{\sigma_i\}} \widetilde{r}^f_j}{e^f_{\sigma_i}}}-\sum_{j\in \mathcal{R}\setminus\{\sigma_i\}} \widetilde{r}^f_j
    = \frac{(|\mathcal{S}_f|-1)C^f}{\sum_{j\in \mathcal{S}_f} e^f_j} -   \frac{(|\mathcal{S}_f|-1)^2 C^f e^f_{\sigma_i}}{[\sum_{j\in \mathcal{S}_f}e^f_j]^2} = \widetilde{r}^f_{\sigma_i}.
\end{align*}

Therefore, given $f$, $\widetilde{\mathbf{r}}^f$ is an NE of $FG(f)$, and consequently, $\widetilde{\mathbf{r}}$ is an NE of the FG-MLMF-SG according to Proposition~\ref{pr:sg_dec}.

\section{Proof of Theorem~\ref{th:uniqu_f_no_lim}}
\label{apdix:uniqu_f_no_lim}
Based on Proposition~\ref{pr:sg_dec}, to prove the uniqueness of the NE of FG-MLMF-SG, it is sufficient to show the uniqueness of the NE of $FG(f)$ for any $f \in \mathcal{F}$.  Let $\hat{\mathcal{S}}_f=\{i\in \mathcal{R}: \hat{r}^f_i > 0\}$.

1) $\hat{r}^f_i=\frac{(|\mathcal{\hat{S}}_f|-1)C^f}{\sum_{j\in \mathcal{\hat{S}}_f}e^f_j} \Big(1-\frac{(|\mathcal{\hat{S}}_f|-1)e^f_i}{\sum_{j\in \mathcal{\hat{S}}_f}e^f_j}\Big)$ if $i\in \hat{\mathcal{S}}_f$; otherwise $\hat{r}^f_i=0$. 
      Considering that $\sum_{j\in \mathcal{R}}\hat{r}^f_j = \sum_{j\in \mathcal{\hat{S}}_f}\hat{r}^f_j$, we obtain from $\frac{\partial V^f_i}{\partial r^f_i}=0$,
\begin{align}
\label{eq:temp1}
    \frac{-C^f \hat{r}^f_i}{[\sum_{j\in \mathcal{\hat{S}}_f}\hat{r}^f_j]^2} +
    \frac{ C^f }{\sum_{j\in \mathcal{\hat{S}}_f}\hat{r}^f_j} -e^f_i = 0,~~~i\in \hat{\mathcal{S}}_f.
\end{align}
Furthermore, we have $|\mathcal{\hat{S}}_f| C^f -C^f = \Big(\sum_{j\in \mathcal{\hat{S}}_f}\hat{r}^f_j\Big) \Big( \sum_{j\in \mathcal{\hat{S}}_f}e^f_j\Big)$ by summing up the left hand side (LHS) of~\eqref{eq:temp1} over all nodes in $\hat{\mathcal{S}}_f$. Therefore, we have
\begin{align}
\label{eq:temp2}
    \sum_{j\in \mathcal{\hat{S}}_f} \hat{r}^f_j =  \frac{(|\mathcal{\hat{S}}_f|-1)C^f}{\sum_{j\in \mathcal{\hat{S}}_f}e^f_j}.
\end{align}
Feeding~\eqref{eq:temp2} back into~\eqref{eq:temp1} and letting $\hat{r}^f_j=0$ for any $j\in \mathcal{R}\setminus \hat{\mathcal{S}}$, we obtain
\begin{align}
\label{eq:temp3}
    \hat{r}^f_i= \frac{(|\mathcal{\hat{S}}_f|-1)C^f}{\sum_{j\in \mathcal{\hat{S}}_f}e^f_j} \Big(1-\frac{(|\mathcal{\hat{S}}_f|-1)e^f_i}{\sum_{j\in \mathcal{\hat{S}}_f}e^f_j}\Big)
\end{align}
for every $i \in \hat{\mathcal{S}}$. This proves 1).

%

2) $i\geq2 $. Assume $i=0$, then any AP, e.g., AP $j$, can increase its utility from 0 to $\frac{C^f}{2}$ by unilaterally changing its offloading data size from 0 to $\frac{C^f}{2e^f_j}$, contradicting the NE assumption and demonstrating $i\geq1$. Now, assume that $i=1$. This means $\hat{r}^f_{\sigma_1}>0$ and $\hat{r}^f_{\sigma_k}=0$ for all $k\in \mathcal{R}\setminus\{1\}$. According to~\eqref{ob_node}, the current utility of AP $\sigma_1$ for flow $f$ is $C^f-\hat{r}^f_{\sigma_1} e^f_{\sigma_1}$. Hence, AP $\sigma_1$ can increase its utility by unilaterally changing the amount of data it offloads, contradicting the NE assumption. Therefore $i\geq2$.

On the other hand, considering the definition of $\mathcal{\hat{S}}_f$, we know that $\hat{r}^f_i >0$ for every $i \in \mathcal{\hat{S}}_f$. From~\eqref{eq:temp3}, $\hat{r}^f_i >0$ implies that $ \frac{(|\hat{\mathcal{S}}_f|-1)e^f_i}{\sum_{j\in \mathcal{\hat{S}}_f}e^f_j}<1$. Therefore, we have
$ e^f_i < \frac{\sum_{j\in \mathcal{\hat{S}}_f}e^f_j}{|\hat{\mathcal{S}}_f|-1}$ for any $i\in\mathcal{\hat{S}}_f$,
which implies that $\max_{i \in \hat{\mathcal{S}}_f} \{ e^f_i \} < \frac{\sum_{j\in \mathcal{\hat{S}}_f}e^f_j}{|\hat{\mathcal{S}}_f|-1}$ for any $i\in\hat{\mathcal{S}}_f$.
%
That is, when APs are ordered such that $e^f_{\sigma_1} \leq e^f_{\sigma_2} \leq \cdots \leq e^f_{\sigma_R}$, $\mathcal{\hat{S}}_f$ is always composed of the APs with the least offloading cost. Assume $\mathcal{\hat{S}}_f=\{\sigma_1,\cdots,\sigma_k\}$ where $e^f_{\sigma_{k+1}} < \frac{\sum^{k}_{j=1} e^f_{\sigma_j} }{k-1}$, we have $\sigma_{k+1}\notin \mathcal{\hat{S}}_f$, and further $\hat{r}^f_{\sigma_{k+1}}=0$. Thus,
\begin{align*}
    \frac{V^f_{\sigma_{k+1}}}{\partial r^f_{\sigma_{k+1}}}\Big|_{r^f_{\sigma_{k+1}}=\hat{r}^f_{\sigma_{k+1}}}=\frac{C^f}{\sum_{j\in \mathcal{\hat{S}}_f} \hat{r}^f_j}-e^f_{\sigma_{k+1}} = \frac{\sum_{j\in \mathcal{\hat{S}}_f}e^f_j}{k-1}-e^f_{\sigma_{k+1}} > 0,
\end{align*}
which implies that AP $\sigma_{k+1}$ can increase its utility by unilaterally increasing its offloading data size, contradicting the NE assumption. Hence, $\mathcal{\hat{S}}_f=\{\sigma_1,\cdots,\sigma_i\}$ and $e^f_{\sigma_{i+1}} \geq \frac{\sum^{i}_{j=1} e^f_{\sigma_j} }{i-1}$.

Statement 1), which gives the optimal amount of offloaded data, and Statement 2), which implies that $\mathcal{\hat{S}}_f$ has a threshold structure concerning the offloading cost, show the uniqueness of the NE of $FG(f)$, and furthermore, the uniqueness of the NE of FG-MLMF-SG is obtained based on Proposition~\ref{pr:sg_dec}.

\section{Proof of Theorem~\ref{th:sym_interior_equi}}
\label{apdix:sym_interior_equi}

We prove this theorem by contradiction. Suppose that in a strictly interior equilibrium $\mathbf{r}$, there exists a flow $f_0$ and APs $i,~j$ such that $r^{f_0}_i < r^{f_0}_j$. Consider the first order derivative
\begin{align}
\label{dvi_dr}
    \frac{\partial V_i}{\partial r^f_i} = \frac{\sum_{k \in \mathcal{R}} r^f_k -r^f_i }{[\sum_{k \in \mathcal{R}}r^f_k]^2}C^f-e^f .
\end{align}

It follows that, for flow $f_0$, we have
\begin{align}
\label{dvi_dr_gl}
    r^{f_0}_i < r^{f_0}_j \Longrightarrow \frac{\sum_{k \in \mathcal{R}} r^{f_0}_k -r^{f_0}_i }{[\sum_{k \in \mathcal{R}} r^{f_0}_k]^2}C^{f_0} - e^{f_0}= \frac{\partial V_i}{\partial r^{f_0}_i} > \frac{\partial V_j}{\partial r^{f_0}_j} =\frac{\sum_{k \in \mathcal{R}} r^{f_0}_k -r^{f_0}_j }{[\sum_{k \in \mathcal{R}} r^{f_0}_k]^2}C^{f_0}- e^{f_0}.
\end{align}

On the other hand, since $\mathbf{r}$ is strictly interior, it follows from KKT condition (A3) that $\nu^f_i=0, \forall i\in \mathcal{R}, f\in \mathcal{F}$. Thus, combing (A1) and~\eqref{dvi_dr_gl}, we have
\begin{align*}
    \lambda_i = \frac{\partial V_i}{\partial r^{f_0}_i} > \frac{\partial V_j}{\partial r^{f_0}_j} = \lambda_j,
\end{align*}
which therefore leads to $r^f_i<r^f_j$ for any $f\in \mathcal{F}$ (not just $f_0$), and, therefore, $\sum_{f} r^f_i<\sum_{f}r^f_j\leq B$ which implies $\lambda_i=0$ according to (A2). Obviously, this contradicts $\lambda_i> \lambda_j\geq0$. This completes the proof of the theorem.

\section{Proof of Theorem~\ref{th:sym_exist}}
\label{apdix:sym_exist}

Note that $g^f(\rho^f)= \frac{\partial V_i}{\partial  r^f_i }\Big|_{r^f_j = \rho^f}$ for the symmetric strategy profile $\{\rho^f, f \in \mathcal{F}\}$. Thus, conditions (B1)-(B3) coincide with the KKT conditions (A1)-(A3) in this case. Accordingly, the set $\{\rho^f, f \in \mathcal{F}\}$ corresponds to a symmetrical NE i.i.f it satisfies conditions (B1)-(B3).

It remains to be shown that there exists a unique combination of $\{\rho^f\}$, $\lambda$, and $\{\nu^f\}$ satisfying conditions (B1)-(B3). To that end, we define the function $W(\mathbf{x})$, where $\mathbf{x}=\big(x^1,x^2,\cdots,x^{|\mathcal{F}|}\big)$, as: $W(\mathbf{x})\triangleq \sum_{f \in \mathcal{F}} \int^{x^f}_{0}g^f(\xi)d \xi$. Consider the following optimization problem:
\begin{align*}
 \max_{\mathbf{x}} W(\mathbf{x})~~~  s.t.~~~ \sum_{f\in \mathcal{F}} x^f \leq B \text{ and } x^f \geq 0,~~\forall f\in \mathcal{F}.
\end{align*}

Since $ W(\mathbf{x})$ is a sum of integrals of decreasing functions, it is continuously differential and concave, and therefore, the above constrained optimization problem over a compact region must have a unique solution, which is denoted by $\{\rho^f, f \in \mathcal{F}\}$. This solution must satisfy the KKT conditions for problem~\eqref{op_node}, which are precisely the conditions listed in (B1)-(B3).

\section{Proof of Lemma~\ref{lemma:prop}}
\label{apdix:prop}

(1) The continuity of $\lambda$ with respect to $C^f$ is immediate from conditions (B1)-(B3) and the continuity of $g^f$. To establish the monotonicity, suppose to the contrary that $\lambda_a=\Lambda(C^{f_0}_a)>\Lambda(C^{f_0}_b)=\lambda_b$ for some $C^{f_0}_a<C^{f_0}_b$. Suppose that $\{\rho^f_a,f\in \mathcal{F}\}$ and $\{\rho^f_b,f\in \mathcal{F}\}$ correspond to the equilibria at $C^f_a$ and $C^f_b$, respectively. Then, $\lambda_a>\lambda_b \geq 0$ implies that $\lambda_a=C^f h^f(\rho^f_a)-e^f >C^f h^f(\rho^f_b)-e^f =\lambda_b$ for all $f\in \mathcal{F}\setminus\{f_0\}$.
Since $h^f(\rho)$ is monotonically decreasing in $\rho$, we have $\rho^f_a<\rho^f_b$ for all $f\in \mathcal{F}\setminus\{f_0\}$. Thus, $\rho^{f_0}_a =B-\sum_{f\in \mathcal{F}\setminus\{f_0\}}\rho^{f}_a >B-\sum_{f\in \mathcal{F}\setminus\{f_0\}}\rho^{f}_b \geq \rho^{f_0}_b$, which implies, according to the decreasing monotonicity of $h^f(\rho)$ and $C^{f_0}_a<C^{f_0}_b$, that $ C^{f_0}_a h^{f_0}(\rho^{f_0}_a)-e^{f_0}=\lambda_a < \lambda_b =C^{f_0}_b h^{f_0}(\rho^{f_0}_b)-e^{f_0}$, which  obviously contradicts $\lambda_a>\lambda_b$.
Thus, we conclude that $\lambda=\Lambda(C^f)$ is continuous and non-decreasing in $C^f$.


(2) The continuity of $\Psi(C^f)$ can be proved in a similar manner as that of $\lambda=\Lambda(C^f)$. We now prove the monotonicity of $\Psi(C^f)$. Assume $0< C^{f_0}_a < C^{f_0}_b $, we consider the following two cases.

Case I. $\Lambda(C^{f_0}_b)=0$. Because of the monotonicity of $\Lambda(C^f)$, we have $\Lambda(C^f_a)=0$ as well. In the equilibria corresponding to $C^f_a$ and $C^f_b$, respectively, we have $C^f_a h^f(\Psi(C^f_a))-e^f=C^f_b h^f(\Psi(C^f_b))-e^f=0$, which implies $\Psi(C^f_a)<\Psi(C^f_b)$ because of the monotonicity of $h^f(\cdot)$.

Case II. $\Lambda(C^{f_0}_b)>0$. We prove the property by contradiction. Suppose to the contrary that $\rho^{f_0}_a=\Psi(C^{f_0}_a)<\Psi(C^{f_0}_b)=f^{f_0}_b$ for some $C^{f_0}_a>C^{f_0}_b$. Suppose that $\{\rho^f_a,f\in \mathcal{F}\}$ and $\{\rho^f_b,f\in \mathcal{F}\}$ correspond to the equilibriums at $C^f_a$ and $C^f_b$, respectively. According to the monotonicity of $\lambda=\Lambda(C^f)$, we have $\lambda_a>\lambda_b \geq 0$, which implies that $\lambda_a=C^f h^f(\rho^f_a)-e^f >C^f h^f(\rho^f_b)-e^f =\lambda_b$ for all $f\in \mathcal{F}\setminus\{f_0\}$.
Since $h^f(\rho)$ is monotonically decreasing in $\rho$, we have $\rho^f_a<\rho^f_b$ for all $f\in \mathcal{F}\setminus\{f_0\}$. Thus, $\rho^{f_0}_a =B-\sum_{f\in \mathcal{F}\setminus\{f_0\}}\rho^{f}_a >B-\sum_{f\in \mathcal{F}\setminus\{f_0\}}\rho^{f}_b \geq \rho^{f_0}_b$, which contradicts $\rho^{f_0}_a<\rho^{f_0}_b$ .
Thus, we conclude that $\rho^f=\Psi(C^f)$ is continuous, and strictly increasing in $C^f \in (0,\infty)$.


(3) From the continuity of $\Psi(C^f)$, it follows that, for any $0 < \rho^f <B$, there exist unique $C^f=\Psi^{-1}(\rho^f)$, $\lambda$, and $\{\rho^{f'},f'\ne f\}$ that construct a symmetric equilibrium together with $\rho^f$. Therefore, these quantities can be regarded as functions of $\rho^f$, and we consider their derivatives with respect to $\rho^f$.

We rewrite condition (B1) for flow $f$ as:
\begin{align}
\label{B1_f}
    C^fh^f(\rho^f)-e^f = \lambda
\end{align}
and, for any flow $f' \in \mathcal{F} \setminus\{ f\}$,
\begin{align}
\label{B1_f'}
    C^{f'} h^{f'}(\rho^{f'})-e^{f'} = \lambda.
\end{align}

Taking the derivative of both sides in~\eqref{B1_f} and~\eqref{B1_f'} with respect to $\rho^f$, respectively, we obtain
\begin{align}
\label{B1_d_f}
   & \frac{d C^f}{d \rho^f}h^f(\rho^f)-C^f \frac{R-1}{R^2}\frac{1}{(\rho^f)^2} = \frac{d \lambda}{d \rho^f} \\
\label{B1_d_f'}
   &  -C^{f'} \frac{R-1}{R^2}\frac{1}{(\rho^{f'})^2} \frac{d \rho^{f'}}{d \rho^f}= \frac{d \lambda}{d \rho^f}
\end{align}

Now, we distinguish two subregions of $\lambda$. If $\rho^f + \sum_{f'\ne f}\rho^{f'} <B$, according to (B2), $\lambda=0$ in the vicinity of $\rho^f$. Thus, $ \frac{d \lambda}{d \rho^f}=0$ and $g^f(\rho^f)=C^f h^f(\rho^f)-e^f=0$. Substituting these results into~\eqref{B1_d_f}, we thus obtain
\begin{align}
\label{C_d_p_eq1}
    \frac{d C^f}{d \rho^f} = \frac{C^f}{\rho^f} =\frac{ e^f R^2}{R-1},
\end{align}
which is non-decreasing in $\rho^f$ with $e^f \geq 0$.

Otherwise,  for $\rho^f + \sum_{f'\ne f}\rho^{f'}=B$, taking the derivative of both sides with respect to $\rho^f$, then $\sum_{f'} \frac{d \rho^{f'}}{d \rho^f}=-1$ in the vicinity of $\rho^f$, combined with~\eqref{B1_d_f'}, which implies
\begin{align*}
   \frac{d \lambda}{d \rho^f}= \sum_{f'\ne f}  C^{f'} \frac{R-1}{R^2}\frac{1}{(\rho^{f'})^2},
\end{align*}
which can be fed back into~\eqref{B1_d_f} to yield
\begin{align}
\label{C_d_p_l1}
\frac{d C^f}{d \rho^f}
=& C^f \frac{R-1}{h^f(\rho^f) R^2}\frac{1}{(\rho^f)^2} + \frac{1}{h^f(\rho^f)} \sum_{f'\ne f}  C^{f'} \frac{R-1}{R^2}\frac{1}{(\rho^{f'})^2} \nonumber \\
=& \frac{(\lambda+e^f)}{[h^f(\rho^f)]^2} \frac{R-1}{ R^2}\frac{1}{(\rho^f)^2} + \frac{1}{h^f(\rho^f)} \sum_{f'\ne f}  C^{f'} \frac{R-1}{R^2}\frac{1}{(\rho^{f'})^2} \nonumber \\
=& (\lambda+e^f)\frac{R^2} {R-1}+ \rho^f \sum_{f'\ne f} \frac{C^{f'} }{(\rho^{f'})^2},
\end{align}
which is increasing in $\rho^f$, since is increasing in $\rho^f$ (Lemma~\ref{lemma:prop}) and $\rho^{f'}$ is decreasing in $\rho^f$.

Combing our findings that both~\eqref{C_d_p_eq1} and~\eqref{C_d_p_l1} are non-decreasing in $\rho^f$, and noticing that the jump in $\frac{d C^f}{d \rho^f}$ at the boundary between the two subregions (namely, the difference between~\eqref{C_d_p_l1} at $\lambda=0$ and~\eqref{C_d_p_eq1}) is positive, we conclude that $C^f=\Psi^{-1}(\rho^f)$ is convex, and, therefore, $\Psi(C^f)$ is concave in the entire range $C^f>0$.

\section{Proof of Theorem~\ref{th:ex_uq}}
\label{apdix:ex_uq}

\textbf{Existence}: Define the mapping $\Phi(\mathbf{C})=\{ \Upsilon^f(\mathbf{C}^{-f}), f \in \mathcal{F}\}$ as the collection of best-response functions to the respective strategy vectors of other flows. Since each component of $\Phi(\mathbf{C})$ is continuous and bounded (Lemma~\ref{lemma:br_bound}), the entire mapping is continuous and bounded. Therefore, it has a fixed point, which is an equilibrium of the leaders' game. This establishes the existence of the SNE.

\textbf{Uniqueness}: The uniqueness of the fixed point requires that, in an equilibrium, $\frac{\partial U_f}{\partial C^f}=0$ must be satisfied for any $f\in \mathcal{F}$. We distinguish two cases in the following. 

If $\rho^f + \sum_{f'\ne f}\rho^{f'} <B$, we have $\frac{d C^f}{d \rho^f} = \frac{C^f}{\rho^f} =\frac{e^f R^2}{R-1} $. Thus,
\begin{align*}
    \frac{\partial U_f}{\partial C^f}
    =& u'_f\Big(R \log (1+\rho^f)\Big) \frac{R}{1+\rho^f}\frac{\partial \rho_f}{\partial C^f} -1 \\
    =& u'_f\Big(R \log (1+\rho^f)\Big) \frac{R-1 }{ R}\frac{1}{ e^f (1+\rho^f) } -1 = 0,
\end{align*}
that is to say,
\begin{align}
\label{contrator_1}
     e^f (1+\rho^f) =  u'_f\Big(R \log (1+\rho^f)\Big) \frac{R-1 }{R}.
\end{align}

Since the LHS of~\eqref{contrator_1} is increasing in $\rho^f$ while the RHS of~\eqref{contrator_1} is decreasing in $\rho^f$, we can conclude that~\eqref{contrator_1} has one solution at most.

On the other hand, if $\rho^f + \sum_{f'\ne f}\rho^{f'}=B$, we have
$\frac{d C^f}{d \rho^f}= (\lambda+e^f)\frac{R^2} {R-1}+ \rho^f \sum_{f'\ne f} \frac{C^{f'} }{(\rho^{f'})^2}$. Thus,
\begin{align*}
    \frac{\partial U_f}{\partial C^f}
    =& u'_f\Big(R \log (1+\rho^f)\Big) \frac{R}{1+\rho^f}\frac{\partial \rho_f}{\partial C^f} -1=0,
\end{align*}
that is to say,
\begin{align}
\label{contrator_2}
(\lambda+e^f)\frac{R^2} {R-1}+ \rho^f \sum_{f'\ne f} \frac{C^{f'} }{(\rho^{f'})^2}  =  u'_f\Big(R \log (1+\rho^f)\Big) \frac{R}{1+\rho^f}.
\end{align}

Similarly, it is easily to see that the LHS of~\eqref{contrator_2} is increasing in $\rho^f$ while the RHS of~\eqref{contrator_2} is decreasing in $\rho^f$. Thus, we can conclude that~\eqref{contrator_2} has one solution at most.

\section{Proof of Lemma~\ref{lemma:br_mono}}
\label{apdix:br_mono}

Given flow $f$, we consider two price vectors $(C^f_a,\mathbf{C}^{-f}_a)$ and $(C^f_b,\mathbf{C}^{-f}_b)$ such that $C^f_a=\Upsilon(\mathbf{C}^{-f}_a)$ and $C^f_b=\Upsilon(\mathbf{C}^{-f}_b)$, and the only difference between $\mathbf{C}^{-f}_a$ and $\mathbf{C}^{-f}_b$ is that one component $C^{f'}, f'\ne f$, is changed between $C^{f'}_a$ and $C^{f'}_b$, where $C^{f'}_a <C^{f'}_b$. The lemma then states that $C^f_a \leq C^f_b$.

We prove the lemma by contradiction. Suppose that $C^f_a > C^f_b$ when $C^{f'}_a <C^{f'}_b$.
If $C^f_b=u_f\Big(R \log (1+B)\Big)$, then the lemma holds trivially since $C^f_b$ is already an upper bound for possible values of $C^f$. Therefore, we assume $C^f_b <u_f\Big(R \log (1+B)\Big)$, i.e., $C^f_b$ is the solution of the equation $\frac{\partial U_f}{\partial C^f}=0$ at $(C^f_b,\mathbf{C}^{-f}_b)$, and further define $\lambda_b$ and $\rho^f_b$ as the respective values of the corresponding followers' equilibrium. Similarly, we assume that $C^f_a$ is the solution of $\frac{\partial U_f}{\partial C^f}=0$ at $(C^f_a,\mathbf{C}^{-f}_a)$, and define $\lambda_a$ and $\rho^f_a$ as the respective values of the corresponding followers' equilibrium.

Next, consider the followers' equilibrium for the price vector $(C^f_b,\mathbf{C}^{-f}_a)$, and denote the respective values by $\lambda_{ba}$ and $\rho^f_{ba}$. For flow $f'$, because of the increasing monotonicity of $\Lambda(C^{f'})$ from Lemma~\ref{lemma:prop}, we conclude $\lambda_{ba}<\lambda_b$ since $C^{f'}_a < C^{f'}_b$. Consequently, for flow $f$, according to condition (B1), we have $C^f_b h^f(\rho^f_{ba})-e^f=\lambda_{ba}<\lambda_b=C^f_b h^f(\rho^f_b)-e^f$, which implies $\rho^f_{ba} > \rho^f_b$ as $h^f(\rho)$ is monotonically decreasing.


Considering the two equilibria for the price vectors $(C^f_a, \mathbf{C}^{-f}_a)$ and $(C^f_b, \mathbf{C}^{-f}_a)$, respectively, we have $\rho^f_a >\rho^f_{ba}$ since $C^f_a > C^f_b$ because of the monotonicity of $\Psi(C^f)$.

Since each of the terms on the RHS of~\eqref{dUf_dC} is decreasing in $\rho^f$ and $U_f$ is concave in $C^f$, $\frac{\partial U_f}{\partial C^f}$ is decreasing in $\rho^f$ and $C^f$, respectively. Then,
for $ \rho^f_b < \rho^f_{ba} <\rho^f_a$ and $C^f_b<C^f_a$, we have
\begin{align*}
   0= \frac{\partial U_f}{\partial C^f} \Big|_{\rho^f_b, C^f_b} >  \frac{\partial U_f}{\partial C^f} \Big|_{\rho^f_{ab}, C^f_b} >   \frac{\partial U_f}{\partial C^f} \Big|_{\rho^f_{a}, C^f_b, } >  \frac{\partial U_f}{\partial C^f} \Big|_{\rho^f_{a}, C^f_a}=0,
\end{align*}
which implies a contradiction. This completes the proof of the lemma.

\section{Proof of Theorem~\ref{th:conv}}
\label{apdix:conv}

First, consider an arbitrary sequence of update steps commencing from an initial vector $\mathbf{C}(0)=(\delta,\delta,\cdots,\delta)$ where $\delta\rightarrow 0^+$, and denote by $\mathbf{\underline{C}}(n)$ the resulting sequence of flow price vector after $n$ updates. Obviously, for any flow $f$, the first time the flow updates its strategy will be a non-decreasing update. In light of Lemma~\ref{lemma:br_mono}, it follows by induction that all updates must be non-decreasing, i.e., $\mathbf{\underline{C}}(n)$ is a non-decreasing sequence. Since $\mathbf{\underline{C}}(n)$ is bounded as well (Lemma~\ref{lemma:br_bound}), it follows that it must converge to a limit. Due to the continuity of the best response function $\Phi(\mathbf{C}^f)$, this limit must be its (unique) fixed point $\mathbf{C}^*$.

In a similar manner, consider a sequence of best-response updates $\mathbf{\overline{C}}(n)$ from an initial vector $\mathbf{C}(0)=\{\eta_1,\cdots,\eta_F\}$ where $\eta_f=u_f\Big(R \log (1+B)\Big)$ (i.e., the upper bounds of the respective flows' best responses). By the same token, Lemma~\ref{lemma:br_mono} implies that all the updates in the sequence must be non-increasing, and the sequence must therefore converge to $\mathbf{C}^*$.

Finally, consider an arbitrary initial vector of flow prices commencing from an arbitrary initial vector of flow prices $\mathbf{C}(0)$. Without loss of generality, assume that all the prices are within the bounds set by Lemma~\ref{lemma:br_bound} (otherwise, consider instead the sequence only after every flow has had at least one opportunity to update its strategy). Then, it follows that $\mathbf{\underline{C}}(n) \leq \mathbf{C}(n) \leq \mathbf{\overline{C}}(n)$ provided that for every $n$ the update step is performed by the same flow in all three sequences. Since, as established above, $\mathbf{\underline{C}}(n) $ and $\mathbf{\overline{C}}(n)$ converge to $\mathbf{C}^*$, it follows that the same is true for $\mathbf{C}(n)$ as well.

\end{document}